\documentclass[12pt]{amsart}
\usepackage{amsmath,mathrsfs}
\usepackage{amsfonts}
\usepackage{amssymb}
\usepackage[all,cmtip]{xy}           
\usepackage{bm}
\usepackage{bbm}
\usepackage{verbatim}
\usepackage{bbding}
\usepackage{txfonts}
\usepackage{tocvsec2}
\usepackage{amscd}
\usepackage{xspace}
\usepackage{graphicx}
\usepackage[thicklines]{cancel}
\usepackage[shortlabels]{enumitem}
\usepackage{ifpdf}

\ifpdf
\usepackage[colorlinks,final,backref=page,hyperindex]{hyperref}
\else
\usepackage[colorlinks,final,backref=page,hyperindex,hypertex]{hyperref}
\fi
\usepackage{tikz}
\usepackage[active]{srcltx}

\usepackage{tikz-cd}
\usepackage{tikz}

\makeatletter

\newtheorem{thm}{Theorem}[section]
\newtheorem{lem}[thm]{Lemma}
\newtheorem{cor}[thm]{Corollary}
\newtheorem{pro}[thm]{Proposition}
\theoremstyle{definition}
\newtheorem{ex}[thm]{Example}
\newtheorem{rmk}[thm]{Remark}
\newtheorem{defi}[thm]{Definition}

\newcommand{\nc}{\newcommand}
\newcommand{\delete}[1]{}

	\nc{\mlabel}[1]{\label{#1}}  
	\nc{\mcite}[1]{\cite{#1}}  
	\nc{\mref}[1]{\ref{#1}}  
	\nc{\meqref}[1]{\eqref{#1}}  
	\nc{\mbibitem}[1]{\bibitem{#1}} 

\delete{
	\nc{\mlabel}[1]{\label{#1}{\hfill \hspace{1cm}{\bf{{\ }\hfill(#1)}}}}
	\nc{\mcite}[1]{\cite{#1}{{\bf{{\ }(#1)}}}}  
	\nc{\mref}[1]{\ref{#1}{{\bf{{\ }(#1)}}}}  
	\nc{\meqref}[1]{\eqref{#1}{{\bf{{\ }(#1)}}}}  
	\nc{\mbibitem}[1]{\bibitem[\bf #1]{#1}} 
}

\newcommand {\emptycomment}[1]{}

\delete{

}

\nc{\oprn}{\theta}

\nc{\calo}{\mathcal{O}}
\nc{\oop}{$\mathcal{O}$-operator\xspace}
\nc{\oops}{$\mathcal{O}$-operators\xspace}
\nc{\mrho}{{\bm{\varrho}}}
\nc{\bfk}{\mathbf{K}}
\nc{\invlim}{\displaystyle{\lim_{\longleftarrow}}\,}
\nc{\ot}{\otimes}

\nc{\desc}{descendant\xspace}
\nc{\eval}[1]{\Big|_{#1}}

\newcommand{\B }{\mathfrak{B}}
\newcommand{\D }{\mathfrak{D}}
\newcommand{\ac}{\triangleright}

\newcommand{\be }{\begin{equation}}
	\newcommand{\ee }{\end{equation}}

\newcommand{\Gr}{\mathrm{Gr}}
\newcommand{\g}{\mathfrak g}
\newcommand{\h}{\mathfrak h}

\newcommand{\Real}{\mathbb R}

\newcommand{\Id}{{\rm{Id}}}

\newcommand{\br}[1]{   [ \cdot,    \cdot  ]   }

\newcommand{\Der}{\mathrm{Der}}

\newcommand{\Aut}{\mathrm{Aut}}

\newcommand{\Ker}{\mathrm{Ker}}

\newcommand{\Img}{\mathrm{Im}}

\newcommand{\Diff}{\mathrm{Diff}}

\nc{\CV}{\mathbf{C}}

\begin{document}
	
\title{Relative Rota-Baxter operators and crossed homomorphisms on Lie 2-groups}
	
\author{Honglei Lang}
\address{College of Science, China Agricultural University, Beijing 100083, China}
\email{hllang@cau.edu.cn}

\author{Shining Wang}
\address{College of Science, China Agricultural University, Beijing 100083, China}
\email{wsn@cau.edu.cn}

\begin{abstract}
A relative Rota-Baxter operator on Lie 2-groups is introduced as a pair of relative Rota-Baxter operators on the underlying Lie groups which is also a Lie groupoid morphism. Such an operator induces a factorization theorem for  Lie 2-groups and gives rise to a  categorical solution of the Yang-Baxter equation.  We further define relative Rota-Baxter operators on Lie group crossed modules. The well-known one-to-one correspondence between Lie 2-groups and crossed modules is extended to an equivalence between the respective relative Rota-Baxter operators on these two structures. Finally, as the formal inverse of relative Rota-Baxter operators, crossed homomorphisms on Lie 2-groups are also studied.
\end{abstract}

\keywords{Relative Rota-Baxter operator, Lie 2-group, Crossed module, Crossed homomorphism}

\maketitle
	
\tableofcontents
	
\allowdisplaybreaks

\section{Introduction}
\settocdepth{part}
\subsection{Relative Rota-Baxter operators and crossed homomorphisms}
The notion of Rota-Baxter operators was first introduced by Baxter for associative algebras and then brought into the areas of analysis and combinatorics by Rota and Cartier; see \cite{G}. Rota-Baxter operators play important roles in Connes-Kreimer’s algebraic framework for renormalization in quantum field theory \cite{CK} and are also connected with noncommutative symmetric functions, Hopf algebras and double algebras \cite{EFMP, GK, Guo, YGT}.  Rota-Baxter operators on Lie algebras first appeared in \cite{Se} which are closely related to the classical Yang-Baxter equation (CYBE). Explicitly, a Rota-Baxter operator of weight 0 is the operator form of a skew-symmetric solution of the CYBE and a Rota-Baxter operator of weight 1 corresponds to a solution of the modified CYBE and leads to a factorization of a Lie algebra. The corresponding Lie group factorization  yields natural solutions of a certain Hamiltonian system. See \cite{Se} for more applications in integrable systems. More generally, relative Rota-Baxter operators (also known as $\mathcal{O}$-operators) on Lie algebras \cite{Bo, K} give rise to solutions of the CYBE in the semi-direct product Lie algebras and induce post-Lie algebra structures \cite{B, BGN1, BGN}.

More recently, Rota-Baxter operators on (Lie) groups were introduced in \cite{GLS}.  The factorization theorem of a Lie group was then realized directly. Subsequently, Rota-Baxter operators on cocommutative Hopf algebras were studied in \cite{Gon}. The general notion of relative Rota-Baxter operators on Lie groups, together with their cohomology theory were further developed in \cite{JSZ}. Moreover, (relative) Rota-Baxter operators not only lead to  pre-Lie or post-Lie group structures, but are also  closely connected with skew left braces, braided groups, and set-theoretical solutions of the Yang-Baxter equation \cite{BGST, BG, ESS, GV}. 

Crossed homomorphisms on groups originated in Whitehead's early work in combinatorial topology. In particular, bijective crossed homomorphisms were applied to provide set-theoretical solutions of the Yang-Baxter equation \cite{LYZ} and to characterize regular subgroups of holomorphs in the classification of Hopf-Galois structures \cite{T}. Crossed homomorphisms can be seen as the formal inverse of relative Rota-Baxter operators. As a special case, differential operators on groups \cite{GLS} are the formal inverse of Rota-Baxter operators. Recently, crossed homomorphisms have played a key role in constructing representations of mapping class groups of surfaces \cite{CS}. In the context of  Lie algebras, crossed homomorphisms were introduced in \cite{L}  for the study of non-abelian extensions of Lie algebras. They have then been used to investigate post-Lie algebras and post-Lie Magnus expansions \cite{MQ} and CKMM theorem for difference Hopf algebras \cite{GLST}. The categorification was explored in \cite{LW}. 

\subsection{Lie 2-groups and Lie group crossed modules} To study symmetries between symmetries, Lie groups were categorifed to Lie 2-groups in \cite{BL}. A (strict) Lie 2-group is a Lie groupoid object in the category of Lie groups, or a Lie group object in the category of Lie groupoids. Lie 2-groups have found broad applications across homotopy theory, topological quantum field theory, gauge field theory, quantum gravity and Lie theory \cite{BW, LL, W, WZ}. An action of a Lie 2-group on a Lie groupoid appeared in the study of left-invariant vector fields on Lie 2-groups \cite{Ler} and has since been studied intensively \cite{GZ, HS, LL}. Notably, Lie 2-groups are equivalent to the much earlier notion of Lie group crossed modules, which was originally introduced by Whitehead to provide a purely algebraic characterization of the second relative homotopy group \cite{W1}. See \cite{Wa} for crossed modules in other algebraic contexts. In \cite{N}, actor crossed modules were proposed as an analogue of automorphism groups of groups and were further used to define actions of crossed modules. The infinitesimal counterparts of Lie 2-groups and Lie group crossed modules are Lie 2-algebras and Lie algebra crossed modules; see \cite{BC} for details. 

\subsection{Main results and outline of the paper} The goal of this paper is to study Rota-Baxter operators on Lie 2-groups. First, inspired by the close relationship among relative Rota-Baxter operators, post-groups, and the Yang-Baxter equation, as established in \cite{BGST},  we aim to extend relative Rota-Baxter operators from Lie groups to Lie 2-groups and investigate their properties generalizing those on Lie groups \cite{BKSZZ, GLS, JSZ}. Then, in \cite{Jiang}, Rota-Baxter operators on Lie group crossed modules were introduced to construct categorical solutions of the Yang-Baxter equation. In light of the correspondence between Lie 2-groups and crossed modules, we intend to generalize Rota-Baxter operators to the relative case and then establish their equivalence to relative Rota-Baxter operators on Lie 2-groups. It turns out that the framework of Lie 2-groups is much easier  and more intrinsic in characterizing properties of relative Rota-Baxter operators.

The paper is organized as follows. In Section \ref{S2}, we recall some basic knowledge of Lie 2-groups, Lie group crossed modules and relative Rota-Baxter operators on Lie groups. In section \ref{S3}, we introduce the notion of relative Rota-Baxter operators on Lie 2-groups  (Definition \ref{rrb on L2 grp}) and characterize relative Rota-Baxter operators by the semi-direct product Lie 2-groups  (Theorem \ref{graph}). Moreover, relative Rota-Baxter operators on Lie 2-groups could induce descendant Lie 2-groups and actions (Theorem \ref{descendt 2grp}). Applications including the factorization theorem of a Lie 2-group (Theorem \ref{F-Thm}) and a construction of categorical solutions of the Yang-Baxter equation (Theorem \ref{solu of YBE}) are provided. Section \ref{S4} is dedicated to the notion of relative Rota-Baxter operators on Lie group crossed modules (Definition \ref{RBO of 2Gp}) and the bijection between relative Rota-Baxter operators on Lie 2-groups and on crossed modules (Theorem \ref{bj}). The infinitesimal counterpart is also investigated  (Theorem \ref{diff}). In Section \ref{S5}, crossed homomorphisms both on Lie 2-groups and on crossed modules are studied. They serve as the formal inverse of relative Rota-Baxter operators (Theorem \ref{inverse}). 
\resettocdepth

\section{Preliminaries}\label{S2}
We begin by reviewing the notions of Lie 2-groups and Lie group crossed modules, as well as the one-to-one correspondence between them. Then we recall relative Rota-Baxter operators on Lie groups along with some basic properties.

A \textbf{Lie groupoid} $P\rightrightarrows P_0$ involves two smooth manifolds $P$ of arrows and $P_0$ of objects, an inclusion $\iota:P_0\to P$, two surjective submersions $s,t:P\to P_0$ called source and target,  and a smooth multiplication $*:P^{(2)}\to P$, where $P^{(2)}=\{(p,p')\in P\times P|sp=tp'\}$, such that 
\begin{itemize}
	\item[\rm(1)] $t\circ\iota=s\circ\iota=\Id_{P_0}$ and $s(p*p')=sp',t(p*p')=tp$;
	\item [\rm(2)] Associativity: $(p*p')*p''=p*(p'*p''), \forall (p,p'),(p',p'')\in P^{(2)}$;
	\item [\rm(3)] Units: $\iota(tp)*p=p=p*\iota(sp)$;
	\item [\rm(4)] Inverses: For all $p\in P$, there exists $q\in P$ satisfying $sp=tq, tp=sq$ and $p*q=\iota(tp), q*p=\iota(sp)$.
\end{itemize}
A \textbf{Lie groupoid morphism} from $P\rightrightarrows P_0$ to $Q\rightrightarrows Q_0$ is a pair of smooth maps $f:P\to Q$ and $f_0:P_0\to Q_0$ such that $s_Q\circ f=f_0\circ s_P, t_Q\circ f=f_0\circ t_P$ and $f(p*p')=f(p)*f(p')$, for all $(p,p')\in P^{(2)}$. For details, see \cite{Mac}.
\begin{defi}[\cite{BL}]
	A (strict) \textbf{Lie 2-group} is a Lie groupoid $P\rightrightarrows P_0$ such that $P$ and $P_0$ are both Lie groups and all the groupoid structure maps are Lie group homomorphisms. 
	
	A \textbf{Lie 2-group homomorphism} from $P\rightrightarrows P_0$ to $Q\rightrightarrows Q_0$ is a Lie groupoid morphism $(f,f_0)$ with $f:P\to Q$ and $f_0:P_0\to Q_0$ such that they are both Lie group homomorphisms.
\end{defi}

Throughout this paper, we always use $\cdot$ to denote group multiplication and $*$ to denote groupoid multiplication.

\begin{defi}[\cite{W1}]
	A \textbf{Lie group crossed module} $(G_1\xrightarrow{\mu}G_0)$ consists of two Lie groups $G_1$ and $G_0$, a Lie group homomorphism $\mu:G_1\to G_0$, and an action $\ac:G_0\to\Aut(G_1)$ of $G_0$ on $G_1$ by automorphisms, such that
	\begin{eqnarray*}
		\mu g_1\ac g_1'=g_1\cdot g_1'\cdot g_1^{-1},\qquad
		\mu(g_0\ac g_1)=g_0\cdot\mu g_1\cdot g_0^{-1}, \qquad \forall g_0\in G_0,g_1,g_1'\in G_1.
	\end{eqnarray*}
\end{defi}

It is standard that Lie 2-groups are in one-to-one correspondence with Lie group crossed modules; see \cite{BL, BS}.
We sketch the correspondence and fix our notations for later use.

Given a Lie 2-group $P\rightrightarrows P_0$, we obtain a crossed module $(\ker s_P\xrightarrow{t_P}P_0)$, where $s_P$ and $t_P$ are the source and target maps. The action $\ac: P_0\to\Aut(\ker s_P)$ is 
\[p_0\ac p=\iota(p_0)\cdot p\cdot\iota(p_0)^{-1},\qquad  \forall p\in\ker s_P, p_0\in P_0.\]
Note that $(\cdot)^{-1}$ stands for the group inverse of an element in $P$.

Conversely, given a crossed module $(G_1\xrightarrow{\mu}G_0)$,  we have a Lie 2-group $G_1\rtimes G_0\rightrightarrows G_0$ whose group multiplication is 
\begin{eqnarray*}
 (g_1,g_0)\cdot(g_1',g_0')=(g_1\cdot (g_0\ac g_1'),g_0\cdot g_0'),
 \end{eqnarray*}
  and groupoid structure is 
  \begin{itemize}
	\item source and target maps: $s(g_1,g_0)=g_0,t(g_1,g_0)=\mu g_1\cdot g_0$;
	\item groupoid multiplication: $(g_1,g_0)*(g_1',g_0')=(g_1\cdot g_1',g_0')$, if $g_0=\mu g_1'\cdot g_0'$.
\end{itemize}

The notion of Rota-Baxter operators on Lie groups was first introduced in \cite{GLS} and then generalized to relative Rota-Baxter operators on Lie groups in \cite{JSZ}. 
\begin{defi}[\cite{GLS, JSZ}]
	Let $G$ and $H$ be two Lie groups with an action $\phi:G\to\Aut(H)$ of $G$ on $H$ by automorphisms.  A \textbf{relative Rota-Baxter operator} on $G$ with respect to the action $\phi$ is a smooth map $\B:H\to G$ such that
	\begin{eqnarray}\label{rrbc}
		\B h\cdot_G \B h'=\B(h\cdot_H\phi(\B h)h'), \qquad\forall h,h'\in H,
	\end{eqnarray}
	or equivalently, 
	\begin{eqnarray*}
		\B(h\cdot_H h')=\B h\cdot_G\B(\phi(\B h)^{-1}h').
	\end{eqnarray*}
\end{defi}
In particular, if $H=G$ and $\phi=\mathrm{Ad}:G\to\Aut(G)$, the adjoint action, then it reduces to a \textbf{Rota-Baxter operator} on $G$, namely, a smooth map $\B:G\to G$ such that
\begin{eqnarray*}
	\B g\cdot\B g'=\B(g\cdot\B g\cdot g'\cdot (\B g)^{-1}), \qquad\forall g,g'\in G.
\end{eqnarray*}
We call $(G, \B)$ a \textbf{Rota-Baxter Lie group}.

With a Lie group action $\phi:G\to\Aut(H)$ of $G$ on another Lie group $H$, we have the semi-direct product Lie group $H\rtimes G$ with multiplication given by
\begin{eqnarray*}
	(h,g)\cdot(h',g')=(h\cdot\phi(g)h',g\cdot g'),\qquad \forall h,h'\in H, g,g'\in G.
\end{eqnarray*}

\begin{pro}[\cite{BKSZZ, JSZ}]\label{gr as sub}
 Let $G$ and $H$ be two Lie groups with an action $\phi:G\to\Aut(H)$ of $G$ on $H$ and $\B:H\to G$ be a smooth map. The following statements are equivalent:
 \begin{itemize}
 	\item [\rm(i)] $\B:H\to G$ is a relative Rota-Baxter operator on $G$ with respect to $\phi$.
 	\item [\rm(ii)] The graph $\Gr(\B)=\{(h,\B h)|\forall h\in H\}$ is a Lie subgroup of the semi-direct product Lie group $H\rtimes G$.
 	\item [\rm(iii)] There is a Rota-Baxter Lie group $(H\rtimes G,\hat{\B})$, where $\hat{\B}:H\rtimes G\to H\rtimes G$ is defined by
 	\begin{eqnarray*}
 		\hat{\B}(h,g)=(e,g^{-1}\cdot \B h), \qquad \forall h\in H, g\in G.
 	\end{eqnarray*}
 \end{itemize}
\end{pro}

\begin{rmk}\label{hatB}
In \cite{BKSZZ}, only  the direction $\rm(i)\Rightarrow\rm(iii)$ in Proposition \ref{gr as sub} is proposed. It is in fact an equivalence. Following from	
\begin{eqnarray*}
		\hat{\B}(h,g)\cdot\hat{\B}(h',g')&=&(e,g^{-1}\cdot\B h\cdot g'^{-1}\cdot\B h'),\\
		\hat{\B}\big((h,g)\cdot\hat{\B}(h,g)\cdot(h',g')\cdot(\hat{\B}(h,g))^{-1}\big)
		&=&\hat{\B}(h\cdot\phi(\B h)h',\B h\cdot g'\cdot(\B h)^{-1}\cdot g)\\
		&=&\big(e,g^{-1}\cdot\B h\cdot g'^{-1}\cdot(\B h)^{-1}\cdot\B(h\cdot\phi(\B h)h')\big),
	\end{eqnarray*}
	we deduce that $\hat{\B}$ is a Rota-Baxter operator on $H\rtimes G$ if and only if $\B$ satisfies \eqref{rrbc}, i.e., it is a relative Rota-Baxter operator on $G$ with respect to $\phi$.
\end{rmk}

\begin{thm}[\cite{GLS, JSZ}]\label{des grp and ac}
	Let $\B:H\to G$ be a relative Rota-Baxter operator on $G$ with respect to an action $\phi:G\to\Aut(H)$.
	\begin{itemize}
	\item[\rm(i)] We have a Lie group $(H,\cdot_\B)$, denoted by $H^\B$ and called the \textbf{descendant Lie group}, where
	\begin{eqnarray*}
		h\cdot_\B h'=h\cdot\phi(\B h)h',\qquad \forall h,h'\in H.
	\end{eqnarray*}
Further, $\B$ is a Lie group homomorphism from $H^\B$ to $G$.
\item[\rm(ii)] There is an action $\Theta:H^\B\to\Diff(G)$ of $H^\B$ on the manifold $G$ defined by
\begin{eqnarray*}
	\Theta(h)g=(\B(\phi(g)h^\dagger))^{-1}\cdot g\cdot\B h^\dagger,
\end{eqnarray*}
where $h^\dagger$ is the inverse of $h$ in  $H^\B$ and $\Diff(G)$ is the group of diffeomorphisms of $G$.
\item[\rm(iii)] In particular, if $(G,\B)$ is a Rota-Baxter Lie group, then $(G^\B,\B)$ is also a Rota-Baxter Lie group.
	\end{itemize}
\end{thm}

\resettocdepth
 
\section{Relative Rota-Baxter operators on Lie 2-groups}\label{S3}
In this section, we introduce the notion of relative Rota-Baxter operators on Lie 2-groups as a categorification of the operators on Lie groups and derive the basic properties parallel to the Lie group case.

\subsection{Definitions and properties}
An action of a Lie 2-group $P\rightrightarrows P_0$ on a Lie groupoid $X\rightrightarrows X_0$ is a pair of group actions $\phi: P\times X\to X$ and $\phi_0: P_0\times X_0\to X_0$ on manifolds such that $(\phi,\phi_0)$ is a Lie groupoid morphism; see \cite{GZ, Ler}. Here we further define  actions of a Lie 2-group on another Lie 2-group.
\begin{defi}
	Let $P\rightrightarrows P_0$ and $Q\rightrightarrows Q_0$ be two Lie 2-groups. A \textbf{Lie 2-group action} $(\phi,\phi_0)$ of $P\rightrightarrows P_0$ on $Q\rightrightarrows Q_0$ consists of two Lie group actions $\phi:P\to\Aut(Q)$ and $\phi_0:P_0\to\Aut(Q_0)$ (by automorphisms) such that the action map $(\phi,\phi_0)$:
	\begin{displaymath}
		\xymatrix@=8ex@R=4ex{
			P\times Q \ar[r]^{\phi} \ar@<-.5ex>[d]_{s\times s}  \ar@<.5ex>[d]^{t\times t} & Q \ar@<-.5ex>[d]_{s} \ar@<.5ex>[d]^{t}
			\\
			P_0\times Q_0 \ar[r]^{\phi_0} & Q_0,
		}
	\end{displaymath}
	is a Lie groupoid morphism, where $P\times Q\rightrightarrows P_0\times Q_0$ is the direct product Lie groupoid. 
\end{defi}

For a Lie 2-group action $(\phi,\phi_0)$, we have
\begin{eqnarray}\label{phi}
	\phi(p* p')(q* q')=(\phi(p)q)*(\phi(p')q'), \qquad \forall (p,p')\in P^{(2)}, (q,q')\in Q^{(2)},
\end{eqnarray}
where $P^{(2)}=\{(p,p')\in P\times P|sp=tp'\}$ is the set of composable pairs, and $Q^{(2)}$ is defined similarly.

\begin{ex}
	 Denote by $\mathrm{Ad}:P\to\Aut(P)$ and $\mathrm{Ad}_0:P_0\to\Aut(P_0)$ the adjoint actions of Lie groups $P$ and $P_0$. The pair $(\mathrm{Ad},\mathrm{Ad}_0)$ is an action of the Lie 2-group $P\rightrightarrows P_0$ on itself, called the \textbf{adjoint action}.
	 
	Note that the left group multiplication of $P\rightrightarrows P_0$ is a Lie 2-group action on its underlying Lie groupoid, rather than on the Lie 2-group $P\rightrightarrows P_0$.
\end{ex}

\begin{defi}\label{rrb on L2 grp}
   Let $(\phi,\phi_0)$ be a Lie 2-group action of   $P\rightrightarrows P_0$ on another Lie 2-group $Q\rightrightarrows Q_0$. A \textbf{relative Rota-Baxter operator} on $P\rightrightarrows P_0$ with respect to the action $(\phi,\phi_0)$ is a pair $(\B,\B_0)$ of smooth maps:
    \begin{displaymath}
   	\xymatrix@=8ex@R=4ex{
   		Q \ar[r]^{\B} \ar@<-.5ex>[d]_{s}  \ar@<.5ex>[d]^{t} & P \ar@<-.5ex>[d]_{s} \ar@<.5ex>[d]^{t}
   		\\
   		Q_0 \ar[r]^{\B_0} & P_0,
   	}
   \end{displaymath}
    such that
    \begin{itemize} 
\item[\rm(1)]  $\B:Q\to P$ is a relative Rota-Baxter operator on the Lie group $P$ with respect to the action $\phi:P\to\Aut(Q)$,
\item[\rm(2)] $\B_0:Q_0\to P_0$ is a relative Rota-Baxter operator on the Lie group $P_0$ with respect to the action $\phi_0:P_0\to\Aut(Q_0)$, 
\item[\rm(3)] $(\B,\B_0)$ is a Lie groupoid morphism  from $Q\rightrightarrows Q_0$ to $P\rightrightarrows P_0$.
    \end{itemize}
Particularly, a relative Rota-Baxter operator $(\B,\B_0)$ is called a \textbf{Rota-Baxter operator} on a Lie 2-group $P\rightrightarrows P_0$ if
\begin{equation*}
	(Q\rightrightarrows Q_0)=(P\rightrightarrows P_0),\qquad (\phi,\phi_0)=(\mathrm{Ad},\mathrm{Ad}_0).
\end{equation*}
Meanwhile, $(P\rightrightarrows P_0,(\B,\B_0))$ is called a \textbf{Rota-Baxter Lie 2-group}.
\end{defi}

\begin{ex}
    A Lie 2-group homomorphism from $Q\rightrightarrows Q_0$ to $P\rightrightarrows P_0$ is a relative Rota-Baxter operator on $P\rightrightarrows P_0$ with respect to the trivial Lie 2-group action.
\end{ex}

\begin{ex}\label{abelian}
	If $(Q\rightrightarrows Q_0)=(V\rightrightarrows V_0)$, a linear Lie 2-group, then $(\B,\B_0)$ is a relative Rota-Baxter operator on the Lie 2-group $P\rightrightarrows P_0$ with respect to an action $(\phi,\phi_0)$ on $(V\rightrightarrows V_0)$ if $\B:V\to P$ and $\B_0:V_0\to P_0$ satisfy 
	\begin{eqnarray*}
		\B v\cdot\B v'=\B(v+\phi(\B v)v'),\qquad
		\B_0 v_0\cdot\B_0 v_0'=\B_0(v_0+\phi_0(\B_0 v_0)v_0'),
	\end{eqnarray*}
and $(\B,\B_0)$ is a Lie groupoid morphism from $V\rightrightarrows V_0$ to $P\rightrightarrows P_0$.
\end{ex}

\begin{ex}
	 A decomposition of a Lie 2-group $P\rightrightarrows P_0$ consists of Lie 2-subgroups $X\rightrightarrows X_0$ and $Y\rightrightarrows Y_0$ such that $P=XY,X\cap Y=\{e_{P}\}$ and $P_0=X_0Y_0,X_0\cap Y_0=\{e_{P_0}\}$. Such a decomposition gives a Rota-Baxter operator $(\B,\B_0)$ on $P\rightrightarrows P_0$ by
	\begin{eqnarray*}
	\B:P\to P,\quad p\mapsto y^{-1}, \quad\forall p=x\cdot y, \qquad \B_0:P_0\to P_0,\quad p_0\mapsto y_0^{-1}, \quad\forall p_0=x_0\cdot y_0.
	\end{eqnarray*}
In fact, it is known that they are  Rota-Baxter operator on Lie groups. Since the groupoid multiplication is a Lie group homomorphism, we have
\begin{eqnarray*}
	\B(p*p')=\B((x\cdot y)*(x'\cdot y'))=\B((x*x')\cdot(y*y'))=(y*y')^{-1}=y^{-1}*y'^{-1}=\B p*\B p'.
\end{eqnarray*}
Thus $(\B,\B_0)$ is a Lie groupoid endomorphism. 
\end{ex}

\begin{ex}\label{2rb}
    A Rota-Baxter Lie group $(P,\B)$ gives two Rota-Baxter Lie 2-groups $(P\rightrightarrows P,(\B,\B))$ and $(P\rtimes P\rightrightarrows P,(\tilde{\B},\B))$. For the latter, $P\rtimes P$ is the semi-direct product Lie group with multiplication induced by the adjoint action of $P$:
    \begin{eqnarray*}
    	(p,l)\cdot(p',l')=(p\cdot l\cdot p'\cdot l^{-1},l\cdot l').
    \end{eqnarray*}
    The source, target and groupoid multiplication of $P\rtimes P\rightrightarrows P$ are 
    \begin{eqnarray*}
    	s(p,l)=l,\qquad t(p,l)=p\cdot l,\qquad (p,p'\cdot l')*(p',l')=(p\cdot p',l').
    \end{eqnarray*}
    The operator $\tilde{\B}:P\rtimes P\to P\rtimes P$ is defined by
    \begin{eqnarray*}
    	\tilde{\B}(p,l)=\big(\B(p\cdot l)\cdot(\B l)^{-1},\B l\big).
    \end{eqnarray*}
\end{ex}

\begin{ex}\label{barB on 2-grp}
	If $(P\rightrightarrows P_0,(\B,\B_0))$ is a Rota-Baxter Lie 2-group, then $(P\rightrightarrows P_0,(\bar{\B},\bar{\B}_0))$ with $(\bar{\B},\bar{\B}_0)$ defined by 
	\begin{eqnarray*}
		\bar{\B}p=p^{-1}\cdot\B(p^{-1}), \qquad \bar{\B}_0p_0=p_0^{-1}\cdot\B_0(p_0^{-1}), \qquad\forall p\in P,p_0\in P_0,
	\end{eqnarray*}
	is also a Rota-Baxter Lie 2-group. Indeed, by \cite[Proposition 2.4]{GLS}, we see that $\bar{\B}$ and $\bar{\B}_0$ are Rota-Baxter operators on Lie groups $P$ and $P_0$, respectively.  Also, since $(\B,\B_0)$ is a groupoid morphism and the groupoid multiplication is a Lie group homomorphism, we have
	\begin{eqnarray*}
		\bar{\B}(p* p')&=&(p* p')^{-1}\cdot\B((p* p')^{-1})=(p^{-1}* p'^{-1})\cdot(\B (p^{-1})* \B (p'^{-1}))\\
		&=&(p^{-1}\cdot\B (p^{-1}))*(p'^{-1}\cdot\B (p'^{-1}))=\bar{\B} p*\bar{\B}p'.
	\end{eqnarray*}
Hence $(\bar{\B},\bar{\B}_0)$ is a Lie groupoid morphism and further a Rota-Baxter operator on $P\rightrightarrows P_0$.
\end{ex}

Now we characterize relative Rota-Baxter operators on Lie 2-groups by the semi-direct product Lie 2-group.
Given a Lie 2-group action $(\phi,\phi_0)$ of $P\rightrightarrows P_0$ on $Q\rightrightarrows Q_0$, we obtain a \textbf{semi-direct product Lie 2-group} $Q\rtimes P\rightrightarrows Q_0\rtimes P_0$ with the direct product groupoid structure
\begin{eqnarray}\label{sd1}
	(q,p)*(q',p')=(q* q',p* p'), \qquad \forall (p,p')\in P^{(2)}, (q,q')\in Q^{(2)},
\end{eqnarray}
and the semi-direct product group structures
\begin{eqnarray}
	\label{sd2}
	(q,p)\cdot(q',p')&=&(q\cdot\phi(p)q',p\cdot p'),\qquad \forall q,q'\in Q, p,p'\in P,\\
	\label{sd3}
	(q_0,p_0)\cdot(q'_0,p_0')&=&(q_0\cdot\phi_0(p_0)q_0',p_0\cdot p'_0),\qquad \forall q_0,q'_0\in Q, p_0,p_0'\in P_0.
\end{eqnarray}
The particular case when $Q\rightrightarrows Q_0$ is a 2-vector space was given in \cite{LL}.

Denote the graphs of smooth maps $\B:Q\to P$ and $\B_0:Q_0\to P_0$ by
\[\Gr(\B)=\{(q,\B q)|\forall q\in Q\}\subset Q\rtimes P, \qquad \Gr(\B_0)=\{(q_0,\B_0q_0)|\forall q_0\in Q_0\}\subset Q_0\rtimes P_0.\]

\begin{thm}\label{graph}
	 Let $(\phi,\phi_0)$ be a Lie 2-group action of $P\rightrightarrows P_0$ on another Lie 2-group $Q\rightrightarrows Q_0$. The following statements are equivalent:
	\begin{itemize}
		\item[\rm(i)] $(\B,\B_0)$ is a relative Rota-Baxter operator on $P\rightrightarrows P_0$ with respect to $(\phi,\phi_0)$;
		\item[\rm(ii)] The graph $\Gr(\B)\rightrightarrows\Gr(\B_0)$ is a Lie 2-subgroup of the semi-direct product Lie 2-group $Q\rtimes P\rightrightarrows Q_0\rtimes P_0$;
		\item[\rm(iii)] There is a Rota-Baxter Lie 2-group $(Q\rtimes P\rightrightarrows Q_0\rtimes P_0,(\hat{\B},\hat{\B}_0))$, where $\hat{\B}:Q\rtimes P\to Q\rtimes P$ and $\hat{\B}_0:Q_0\rtimes P_0\to Q_0\rtimes P_0$ are defined by
		\begin{eqnarray*}
			\hat{\B}(q,p)=(e,p^{-1}\cdot\B q),\qquad
			\hat{\B}_0(q_0,p_0)=(e,p_0^{-1}\cdot\B_0q_0).
		\end{eqnarray*}
	\end{itemize}
\end{thm}
\begin{proof}
	$\rm(i)\Leftrightarrow\rm(ii)$. We infer from \cite[Proposition 3.4]{JSZ} that $\Gr(\B)$ and $\Gr(\B_0)$ are Lie subgroups of $Q\rtimes P$ and $Q_0\rtimes P_0$ if and only if $\B$ and $\B_0$ are relative Rota-Baxter operators on Lie groups with respect to the actions $\phi$ and $\phi_0$. Then, for $(q,q')\in Q^{(2)}$, observing that	
\begin{eqnarray*}
		(q,\B q)*(q',\B q')=(q* q',\B q* \B q')\in\Gr(\B)
	\end{eqnarray*}
        if and only if $\B q*\B q'=\B(q* q')$, we see that 
 $\Gr(\B)\rightrightarrows\Gr(\B_0)$ is a Lie subgroupoid of $Q\rtimes P\rightrightarrows Q_0\rtimes P_0$ if and only if $(\B,\B_0)$ is a Lie groupoid morphism.

   	$\rm(i)\Leftrightarrow\rm(iii)$. By \cite[Proposition 3.1]{BKSZZ} and Remark \ref{hatB}, $\B$ and $\B_0$ are relative Rota-Baxter operators on Lie groups $P$ and $P_0$ if and only if $\hat{\B}$ and $\hat{\B}_0$ are Rota-Baxter operators on Lie groups $ Q\rtimes P $ and $Q_0\rtimes P_0$. Moreover, for $(q,q')\in Q^{(2)},(p,p')\in P^{(2)}$, we have
   \begin{eqnarray*}
   	\hat{\B}((q,p)*(q',p'))&=&\hat{\B}(q* q',p* p')=(e,(p* p')^{-1}\cdot\B(q* q')),\\
   	\hat{\B}(q,p)*\hat{\B}(q',p')&=&(e,(p^{-1}\cdot\B q)* (p'^{-1}\cdot\B q'))=(e,(p^{-1}* p'^{-1})\cdot(\B q*\B q')),
   \end{eqnarray*}
   which means that $\hat{\B}((q,p)*(q',p'))=\hat{\B}(q,p)*\hat{\B}(q',p')$ if and only if $\B(q* q')=\B q*\B q'$. Therefore, $(\B,\B_0)$ is a relative Rota-Baxter operator on $P\rightrightarrows P_0$ with respect to $(\phi,\phi_0)$ if and only if $(\hat{\B},\hat{\B}_0)$ is a Rota-Baxter operator on the semi-direct product Lie 2-group.
\end{proof}

\begin{thm}\label{descendt 2grp}
    If $(\B,\B_0)$ is a relative Rota-Baxter operator on $P\rightrightarrows P_0$ with respect to an action $(\phi,\phi_0)$ on $Q\rightrightarrows Q_0$, then 
    \begin{itemize}
        \item[\rm(i)] $Q^\B\rightrightarrows Q_0^\B$ is a Lie 2-group with the Lie group structures given by
        \begin{eqnarray}\label{des}
            q\cdot_\B q'=q\cdot\phi(\B q)q',  \qquad
            q_0\cdot_{\B_0} q_0'=q_0\cdot\phi_0(\B_0q_0)q_0', \qquad   \forall q,q'\in Q, q_0,q_0'\in Q_0,  
        \end{eqnarray}
        and the Lie groupoid structure the same as $Q\rightrightarrows Q_0$. We call it the {\bf descendant Lie 2-group}.       Moreover, $(\B,\B_0)$ is a Lie 2-group homomorphism from $Q^\B\rightrightarrows Q_0^\B$ to $P\rightrightarrows P_0$.
        \item[\rm(ii)] There is an action $(\phi^\B,\phi_0^\B)$ of the Lie 2-group $Q^\B\rightrightarrows Q_0^\B$ on the Lie groupoid $P\rightrightarrows P_0$:
         \begin{displaymath}
    \xymatrix@=8ex@R=4ex{
    Q^\B\times P \ar[r]^{\phi^\B} \ar@<-.5ex>[d]_{s\times s}  \ar@<.5ex>[d]^{t\times t} & P \ar@<-.5ex>[d]_{s} \ar@<.5ex>[d]^{t}
    \\
     Q^\B_0\times P_0 \ar[r]^{\phi_0^\B} &  P_0,
    }
\end{displaymath}
which is defined by
        \begin{eqnarray*}
            \phi^\B(q)p=\big(\B(\phi(p)q^\dagger)\big)^{-1}\cdot p\cdot\B q^\dagger,\qquad
            \phi_0^\B(q_0)p_0=\big(\B_0(\phi_0(p_0)q_0^\dagger)\big)^{-1}\cdot p_0\cdot\B_0q_0^\dagger,
        \end{eqnarray*}
    where $q^\dagger$ and $q_0^\dagger$ are the inverses in the Lie groups $Q^\B$ and $Q_0^\B$, respectively.
    \item[\rm(iii)] In particular, if $(P\rightrightarrows P_0,(\B,\B_0))$ is a Rota-Baxter Lie 2-group , then $(P^\B\rightrightarrows P_0^\B,(\B,\B_0))$ is also a Rota-Baxter Lie 2-group.
    \end{itemize}
\end{thm}
\begin{proof}
   \rm(i). According to \cite[Proposition 3.5]{JSZ} as recalled in Theorem \ref{des grp and ac}, we have two descendant Lie groups $Q^\B$ and $Q^\B_0$, as well as two Lie group homomorphisms $\B: Q^\B\to P$ and $\B_0:Q_0^\B\to P_0$. Then, to see that $Q^\B\rightrightarrows Q_0^\B$ is a Lie 2-group, we have
   \begin{eqnarray*}
   	s(q\cdot_\B q')=sq\cdot s\phi(\B q)q'=sq\cdot\phi_0(\B_0 sq)sq'=sq\cdot_{\B_0} sq',
   \end{eqnarray*}
   where we used the facts that $Q\rightrightarrows Q_0$ is a Lie 2-group and $\B_0\circ s_Q=s_P\circ\B$. Similarly, we have $t(q\cdot_\B q')=tq\cdot_{\B_0} tq'$. 
  For the groupoid multiplication, given \eqref{phi} and the fact that $(\B,\B_0)$ is a groupoid morphism, it follows that
  \begin{eqnarray*}
  	(q* q')\cdot_\B(j* j')&=&(q* q')\cdot\phi(\B q*\B q')(j* j')\\
  	&=&(q* q')\cdot(\phi(\B q)j*\phi(\B q')j')\\
  	&=&(q\cdot\phi(\B q)j)*(q'\cdot\phi(\B q')j')\\
  	&=&(q\cdot_\B j)*(q'\cdot_\B j'),
  \end{eqnarray*}
   which implies that $Q^\B\rightrightarrows Q_0^\B$ is a Lie 2-group. 
  It is obvious that $(\B,\B_0)$ is a Lie 2-group homomorphism from $Q^\B\rightrightarrows Q_0^\B$ to $P\rightrightarrows P_0$.
    
    \rm(ii). It follows from  \cite[Theorem 3.8]{JSZ} that $\phi^\B$  and $\phi_0^\B$  are Lie group actions of $Q^\B$ and $Q_0^{\B}$  on the manifolds $P$ and $P_0$, respectively. It is left to show that $(\phi^\B,\phi_0^\B)$ is a groupoid morphism. 
    Based on  \eqref{phi} and the fact that $(\B,\B_0)$ is a Lie groupoid morphism, we have
    \begin{eqnarray*}
    	\phi^\B(q* q')(p* p')&=&
        \big(\B(\phi(p* p')(q^\dagger* q'^\dagger))\big)^{-1}\cdot(p* p')\cdot\B(q^\dagger* q'^\dagger)\\
        &=&(\B\phi(p)q^\dagger*\B\phi(p')q'^\dagger)^{-1}\cdot(p* p')\cdot(\B q^\dagger*\B q'^\dagger)\\
        &=&\big((\B\phi(p)q^\dagger)^{-1}\cdot p\cdot\B q^\dagger\big)*\big((\B\phi(p')q'^\dagger)^{-1}\cdot p'\cdot\B q'^\dagger\big)\\
        &=&\phi^\B(q)p* \phi^\B(q')p'.
    \end{eqnarray*}
    The verification of the other conditions of a Lie groupoid morphism is omitted. Therefore, $(\phi^\B,\phi_0^\B)$ is an action of the Lie 2-group $Q^\B\rightrightarrows Q_0^\B$ on the Lie groupoid $P\rightrightarrows P_0$.
    
   \rm(iii) is a consequence of the facts that $(P^\B,\B)$ and $(P_0^\B,\B_0)$ are Rota-Baxter Lie groups \cite[Proposition 2.13]{GLS} and the underlying Lie groupoid structure of the descendant Lie 2-group $P^\B\rightrightarrows P_0^\B$ is the same as $P\rightrightarrows P_0$.
\end{proof}

The following construction yields more relative Rota-Baxter operators from a known one via Lie 2-group automorphisms. 

\begin{pro}\label{iso rRB}
    Let $(\B,\B_0)$ be a relative Rota-Baxter operator on a Lie 2-group $P\rightrightarrows P_0$ with respect to an action $(\phi,\phi_0)$ on another Lie 2-group $Q\rightrightarrows Q_0$.  If $\theta:Q\to Q,\theta_0:Q_0\to Q_0,\rho:P\to P,\rho_0:P_0\to P_0$ are smooth maps such that $(\theta\times\rho,\theta_0\times\rho_0)$ is a Lie 2-group automorphism of $Q\rtimes P\rightrightarrows Q_0\rtimes P_0$, then
    \begin{eqnarray*}
    	(\rho^{-1}\circ\B\circ\theta,\rho_0^{-1}\circ\B_0\circ\theta_0)
    \end{eqnarray*} 
 is also a relative Rota-Baxter operator on $P\rightrightarrows P_0$ with respect to the action $(\phi,\phi_0)$, where $(\theta\times \rho)(q,p)=(\theta q,\rho p)$. 
\end{pro}
\begin{proof}
     Since $\theta\times\rho$ is an automorphism of the semi-direct product Lie group $Q\rtimes P$, i.e.
     \begin{eqnarray*}
     	(\theta\times\rho)\big((q,p)\cdot(q',p')\big)=(\theta\times\rho)(q,p)\cdot(\theta\times\rho)(q',p'),
     \end{eqnarray*}
we deduce that $\theta\in\Aut(Q)$, $\rho\in\Aut(P)$, and they satisfy $\phi(\rho p)\circ\theta=\theta\circ\phi(p)$ for any $p\in P$. Combining this with the relation \eqref{rrbc} of $\B:Q\to P$, we have
     \begin{eqnarray*}
         (\rho^{-1}\B\theta)q\cdot(\rho^{-1}\B\theta)q'&=&\rho^{-1}(\B\theta q\cdot\B\theta q')=\rho^{-1}\B(\theta q\cdot\phi(\B\theta q)\theta q')\\
         &=&\rho^{-1}\B(\theta q\cdot \theta(\phi(\rho^{-1}\B\theta q)q'))=\rho^{-1}\B\theta(q\cdot\phi(\rho^{-1}\B\theta q)q'),
     \end{eqnarray*}
     which implies that $\rho^{-1}\B\theta$ is a relative Rota-Baxter operator on the Lie group $P$ with respect to $\phi:P\to\Aut(Q)$. Similarly, $\rho_0^{-1}\B_0\theta_0$ is a relative Rota-Baxter operator on the Lie group $P_0$ with respect to $\phi_0:P_0\to\Aut(Q_0)$. The fact that $(\theta\times\rho,\theta_0\times\rho_0)$ is a groupoid endomorphism of $Q\rtimes P\rightrightarrows Q_0\rtimes P_0$ ensures that $(\theta,\theta_0)$ and $(\rho,\rho_0)$ are groupoid endomorphisms of $Q\rightrightarrows Q_0$ and $P\rightrightarrows P_0$. Then the composition $(\rho^{-1}\circ\B\circ\theta,\rho_0^{-1}\circ\B_0\circ\theta_0)$ is a groupoid morphism since $(\B,\B_0)$ is also a groupoid morphism. 
\end{proof}

\begin{cor}\label{corRB}
    If $(P\rightrightarrows P_0,(\B,\B_0))$ is a Rota-Baxter Lie 2-group and $(\rho,\rho_0)$ is a Lie 2-group automorphism of  $P\rightrightarrows P_0$, then $\big(P\rightrightarrows P_0,(\rho^{-1}\circ\B\circ\rho,\rho_0^{-1}\circ\B_0\circ\rho_0)\big)$ is also a Rota-Baxter Lie 2-group.
\end{cor}

\subsection{Factorization theorem of Rota-Baxter Lie 2-groups}
In this subsection, we obtain a factorization theorem for Rota-Baxter Lie 2-groups by adapting the approaches in \cite{GLS, Se} from Lie algebras and Lie groups to the context of Lie 2-groups. 

Recall that a normal Lie 2-subgroup of $P\rightrightarrows P_0$ is a Lie subgroupoid $P'\rightrightarrows P_0'$ such that $P'$ and $P_0'$ are normal Lie subgroups of $P$ and $P_0$. 
 
Suppose that $\big(P\rightrightarrows P_0,(\B,\B_0)\big)$ is a Rota-Baxter Lie 2-group. By Theorem \ref{descendt 2grp}, $(\B,\B_0)$ is a Lie 2-group homomorphism from the descendant Lie 2-group $P^\B\rightrightarrows P_0^\B$ to $P\rightrightarrows P_0$. Actually, we have another homomorphism.

\begin{pro}
	Let $\big(P\rightrightarrows P_0,(\B,\B_0)\big)$ be a Rota-Baxter Lie 2-group. Define two maps
	\begin{eqnarray*}
		\B^+:P\to P, \qquad \B^+p=p\cdot\B p, \qquad \B_0^+:P_0\to P_0, \qquad \B_0^+p_0=p_0\cdot\B_0 p_0.
	\end{eqnarray*}
	Then $(\B^+,\B_0^+)$ is a Lie 2-group homomorphism from $P^\B\rightrightarrows P_0^\B$ to $P\rightrightarrows P_0$.
\end{pro}
\begin{proof}
	It follows from \cite[Proposition 3.1]{GLS} that $\B^+$ and $\B_0^+$ are Lie group homomorphisms. Since $(\B,\B_0)$ is a Lie groupoid morphism and the groupoid multiplication is a Lie group homomorphism, we have
	\begin{eqnarray*}
		\B^+(p* p')=(p* p')\cdot(\B p* \B p')=(p\cdot\B p)*(p'\cdot\B p')=\B^+p* \B^+p'.
	\end{eqnarray*}
	The other conditions for $(\B^+,\B_0^+)$ to be a Lie groupoid morphism also hold.
\end{proof}

Given a Rota-Baxter Lie 2-group $\big(P\rightrightarrows P_0,(\B,\B_0)\big)$, we define four subgroupoids of $P\rightrightarrows P_0$:
\begin{eqnarray*}
	\Img\B^+\rightrightarrows\Img\B_0^+,\qquad \Img\B\rightrightarrows\Img\B_0,\qquad \Ker\B\rightrightarrows\Ker\B_0,\qquad \Ker\B^+\rightrightarrows\Ker\B_0^+.
\end{eqnarray*}
Denote them by $P^+\rightrightarrows P_0^+$, $P^-\rightrightarrows P_0^-$, $K^+\rightrightarrows K_0^+$ and $K^-\rightrightarrows K_0^-$, respectively. As $(\B,\B_0)$ and $(\B^+,\B_0^+)$ are both Lie 2-group homomorphisms, it follows that $P^\pm\rightrightarrows P_0^\pm$ are Lie 2-subgroups of $P\rightrightarrows P_0$ and that $K^\pm\rightrightarrows K_0^\pm$ are normal Lie 2-subgroups of $P^\B\rightrightarrows P_0^\B$. In addition, the following relations hold.
\begin{lem}
$K^\pm\rightrightarrows K_0^\pm$ are normal Lie 2-subgroups of $P^\pm\rightrightarrows P_0^\pm$. 
\end{lem}
\begin{proof}
 By \cite[Lemma 3.2]{GLS}, $K^\pm\subset P^\pm$ and $K_0^\pm\subset P_0^\pm$ are normal Lie subgroups. The fact that $(\B,\B_0)$ and $(\B^+,\B_0^+)$ are both Lie groupoid morphisms ensures that $K^\pm\rightrightarrows K_0^\pm$ are subgroupoids of $P^\pm\rightrightarrows P_0^\pm$.  Thus $K^\pm\rightrightarrows K_0^\pm$ are normal Lie 2-subgroups of $P^\pm\rightrightarrows P_0^\pm$.
\end{proof}

\begin{pro}
	There is a Lie 2-group isomorphism  between the two quotient Lie 2-groups:
	\begin{displaymath}
		\xymatrix@=8ex@R=4ex{
			P^-/K^- \ar[r]^{\Theta} \ar@<-.5ex>[d]  \ar@<.5ex>[d] & P^+/K^+
			\ar@<-.5ex>[d] \ar@<.5ex>[d]
			\\
			P_0^-/K_0^- \ar[r]^{\Theta_0} & P_0^+/K_0^+,
		}
	\end{displaymath}
where $\Theta$ and $\Theta_0$ are defined by
	\begin{eqnarray*}
		\Theta(\overline{\B p})=\overline{\B^+ p},\qquad \Theta_0(\overline{\B_0 p_0})=\overline{\B_0^+ p_0}.
	\end{eqnarray*}
	We call $(\Theta,\Theta_0)$ the \textbf{Cayley transform} of the Rota-Baxter operator $(\B,\B_0)$.
\end{pro}
\begin{proof}
It is known from \cite[Proposition 3.3]{GLS} that $\Theta$ and $\Theta_0$ are well-defined Lie group isomorphisms. To see that $(\Theta,\Theta_0)$ is a groupoid morphism, we only check that $\Theta$ preserves the groupoid multiplication. Since $(\B,\B_0)$ and $(\B^+,\B_0^+)$ are both groupoid morphisms, we have
	\begin{eqnarray*}
		\Theta(\overline{\B p}*\overline{\B p'})=\Theta(\overline{\B(p* p')})=\overline{\B^+(p* p')}=\overline{\B^+p*\B^+p'}=\Theta(\overline{\B p})*\Theta(\overline{\B p'}).
	\end{eqnarray*}
	Therefore, $(\Theta,\Theta_0)$ is a  Lie 2-group isomorphism.
\end{proof}

Now we consider the direct product Lie 2-group $P^+\times P^-\rightrightarrows P_0^+\times P_0^-$. Define 
\begin{eqnarray*}
	P^\Theta=\{(p^+,p^-)\in P^+\times P^-|\Theta(\overline{p^-})=\overline{p^+}\},\qquad
	P_0^\Theta=\{(p_0^+,p_0^-)\in P_0^+\times P_0^-|\Theta_0(\overline{p_0^-})=\overline{p_0^+}\}.
\end{eqnarray*}

\begin{lem}\label{3.14}
	With the above notations, $P^\Theta\rightrightarrows P_0^\Theta$ is a Lie 2-subgroup of $P^+\times P^-\rightrightarrows P_0^+\times P_0^-$ and it is isomorphic to $P^\B\rightrightarrows P_0^\B$ by $(\sigma,\sigma_0)$, where
	\begin{eqnarray*}
		\sigma:P^\B\to P^\Theta,\qquad\sigma(p)=(\B^+ p,\B p),\qquad \sigma_0:P_0^\B\to P_0^\Theta,\qquad  \sigma_0(p_0)=(\B_0^+ p_0,\B _0p_0).
	\end{eqnarray*}
\end{lem}
\begin{proof}
	Since $\Theta$ and $\Theta_0$ are Lie group isomorphisms, we have that $P^\Theta\subset P^+\times P^-$ and $P_0^\Theta\subset P_0^+\times P_0^-$ are Lie subgroups and that $\sigma,\sigma_0$ are Lie group isomorphisms by \cite[Lemma 3.4]{GLS}. It is evident that $P^\Theta\rightrightarrows P_0^\Theta$ is further a subgroupoid since $(\Theta,\Theta_0)$ is a groupoid morphism. Hence $P^\Theta\rightrightarrows P_0^\Theta$ is a Lie 2-subgroup. To show that $(\sigma,\sigma_0)$ is a groupoid morphism, we only check that $\sigma$ preserves the groupoid multiplication. Actually, using the fact that $(\B^+,\B_0^+)$ and $(\B,\B_0)$ are groupoid morphisms and the groupoid multiplication is a Lie group homomorphism, we have
	\begin{eqnarray*}
		\sigma(p* p')&=&(\B^+(p* p'),\B(p* p'))=\big((p* p')\cdot(\B p*\B p'),\B p*\B p'\big)\\
		&=&\big((p\cdot\B p)*(p'\cdot\B p'),\B p*\B p'\big)=\sigma(p)*\sigma(p').
	\end{eqnarray*}
Thus $(\sigma,\sigma_0)$ is a Lie 2-group isomorphism.
\end{proof}

\begin{thm}\label{F-Thm}
	Let $(P\rightrightarrows P_0,(\B,\B_0))$ be a Rota-Baxter Lie 2-group. Then every element $p\in P$ and $p_0\in P_0$ can be uniquely expressed as
	\begin{eqnarray*}
		p=p^+\cdot(p^-)^{-1}, \qquad \text{for } (p^+,p^-)\in P^\Theta,\qquad p_0=p_0^+\cdot(p_0^-)^{-1}, \qquad \text{for } (p_0^+,p_0^-)\in P_0^\Theta.
	\end{eqnarray*}
Moreover, the pair of maps
\begin{eqnarray*}
	p\mapsto (p^+,p^-),\qquad 	p_0\mapsto (p_0^+,p_0^-),
\end{eqnarray*}
is a Lie groupoid morphism from $P\rightrightarrows P_0$ to $P^\Theta\rightrightarrows P_0^\Theta$.
\end{thm}
\begin{proof}
	Since $(P,\B)$ and $(P_0,\B_0)$ are Rota-Baxter Lie groups, we have these unique expressions from \cite[Theorem 3.5]{GLS}. The pair of maps is a Lie groupoid morphism by Lemma \ref{3.14}.
\end{proof}

\subsection{Categorical solutions of the Yang-Baxter equation} 
Recall that for a set $X$, a \textbf{set-theoretical solution} of the Yang-Baxter equation on $X$ is a bijective map $R:X\times X\to X\times X$  satisfying:
\begin{eqnarray}\label{YBE}
	(R\times \Id_X)(\Id_X\times R)(R\times \Id_X)=(\Id_X\times R)(R\times \Id_X)(\Id_X\times R).
\end{eqnarray}

As stated in \cite{BGST}, relative Rota-Baxter operators on Lie groups naturally induce  post-groups and provide set-theoretical solutions of the Yang-Baxter equation. Furthermore, Rota-Baxter operators on crossed modules are used in \cite{Jiang} to construct  \textbf{categorical solutions} of the Yang-Baxter equation for a small category $\mathcal{C}$, which is an invertible functor 
\begin{eqnarray*}
	\mathcal{R}:\mathcal{C}\times \mathcal{C}\to \mathcal{C}\times \mathcal{C}
\end{eqnarray*}
  satisfying \eqref{YBE}. Combining these two ideas, we  find categorical solutions of the Yang-Baxter equation by relative Rota-Baxter operators on Lie 2-groups. 
  
  Denote a Lie 2-group $Q\rightrightarrows Q_0$ by $\mathcal{Q}$ for brevity.
\begin{thm}\label{solu of YBE}
	Let $(\B,\B_0)$ be a relative Rota-Baxter operator on a Lie 2-group $P\rightrightarrows P_0$ with respect to an action $(\phi,\phi_0)$ of $P\rightrightarrows P_0$ on $Q\rightrightarrows Q_0$. Then $\mathcal{R}=(R_\B,R_{\B_0}):\mathcal{Q}\times \mathcal{Q}\to \mathcal{Q}\times \mathcal{Q}$ is a categorical solution of the Yang-Baxter equation \eqref{YBE}, where $R_\B:Q\times Q\to Q\times Q$ and $R_{\B_0}:Q_0\times Q_0\to Q_0\times Q_0$ are defined by
	\begin{eqnarray*}
		R_\B(q,j)&=&\big(\phi(\B q)j,(\phi(\B q)j)^\dagger\cdot_\B q\cdot_\B j\big), \qquad \forall q,j\in Q,\\
		R_{\B_0}(q_0,j_0)&=&\big(\phi_0(\B_0 q_0)j_0,(\phi_0(\B_0 q_0)j_0)^\dagger\cdot_{\B_0} q_0\cdot_{\B_0}j_0\big), \qquad \forall q_0,j_0\in Q_0.
	\end{eqnarray*}
Here $\cdot_{\B}$ and $\cdot_{\B_0}$ are the descendant group structures given by \eqref{des} and $(\cdot)^\dagger$ denotes the corresponding inverse operations.
\end{thm}
\begin{proof}
	By \cite[Corollary 3.14]{BGST}, $R_\B$ and $R_{\B_0}$ are both set-theoretical solutions of the Yang-Baxter equation \eqref{YBE}. It remains to prove that $\mathcal{R}:\mathcal{Q}\times \mathcal{Q}\to \mathcal{Q}\times \mathcal{Q}$ is a functor (groupoid morphism). Since $(\phi,\phi_0)$ is a Lie 2-group action and $Q^\B\rightrightarrows Q_0^\B$ is a Lie 2-group, we have
	\begin{eqnarray*}
		R_{\B_0}(s\times s)(q,j)&=&(\phi_0(\B_0 sq)sj,(\phi_0(\B_0 sq)sj)^\dagger\cdot_{\B_0}sq\cdot_{\B_0}sj)\\
		&=&\big(s\phi(\B q)j,s((\phi(\B q)j)^\dagger\cdot_{\B}q\cdot_{\B}j)\big)\\
		&=&(s\times s)R_\B(q,j),
	\end{eqnarray*}
where we also applied the fact that $s_P\circ\B=\B_0\circ s_Q$.
	It is similar to show that $ R_{\B_0}\circ(t\times t)=(t\times t)\circ R_\B$ and $R_\B\circ(\iota\times\iota)=(\iota\times\iota)\circ R_{\B_0}$.
	By \eqref{phi} and the fact that $Q^\B\rightrightarrows Q_0^\B$ is a Lie 2-group, we obtain
	\begin{eqnarray*}
		R_\B\big((q,j)*(q',j')\big)&=&\big(\phi(\B q*\B q')(j* j'),(\phi(\B q*\B q')(j* j'))^\dagger\cdot_{\B}(q* q')\cdot_{\B}(j* j')\big)\\
		&=&\big(\phi(\B q)j*\phi(\B q')j',(\phi(\B q)j)^\dagger*(\phi(\B q')j')^\dagger\cdot_{\B}((q\cdot_{\B}j)*(q'\cdot_{\B}j'))\big)\\
		&=&\big(\phi(\B q)j*\phi(\B q')j',((\phi(\B q)j)^\dagger\cdot_{\B}q\cdot_{\B}j)*((\phi(\B q')j')^\dagger\cdot_{\B}q'\cdot_{\B}j')\big)\\
		&=&R_\B\big(q,j)* R_\B(q',j'),
	\end{eqnarray*}
	which completes the proof.
\end{proof}

\section{Relative Rota-Baxter operators on crossed modules}\label{S4}
Motivated by the one-to-one correspondence between Lie 2-groups and Lie group crossed modules, we study relative Rota-Baxter operators on crossed modules, extending the notion of Rota-Baxter operators on crossed modules introduced in \cite{Jiang}. We 
 further establish the equivalence between relative Rota-Baxter operators on Lie 2-groups and on crossed modules.  The comparison highlights the conceptual advantage of formulating relative Rota-Baxter operators directly on Lie 2-groups.

\subsection{Equivalence of relative Rota-Baxter operators  on Lie 2-groups and crossed modules}
Let us first recall the actions of a crossed module on another crossed module following \cite{N}.

A \textbf{crossed module morphism} from $(G_1\xrightarrow{\mu}G_0)$ to $(H_1\xrightarrow{\partial}H_0)$ is a pair of Lie group homomorphisms  $\epsilon: G_1\to H_1$ and $\rho:G_0\to H_0$, which satisfies  
\begin{eqnarray*}
	\partial\circ\epsilon=\rho\circ\mu,\qquad 	\epsilon(g_0\ac_G g_1)=\rho(g_0)\ac_H\epsilon(g_1), \qquad \forall g_0\in G_0, g_1\in G_1.
\end{eqnarray*} 
If $\epsilon$ and $\rho$ are both isomorphisms, then $(\epsilon,\rho)$ is called an \textbf{isomorphism}. Denote by $\Aut(H_1,H_0,\partial)$ the automorphism group of $(H_1\xrightarrow{\partial}H_0)$.

For a crossed module $(H_1\xrightarrow{\partial}H_0)$, let  $\Der(H_0,H_1)$ be the set of all derivations from $H_0$ to $H_1$, i.e., smooth maps $\gamma:H_0\to H_1$ satisfying
\begin{eqnarray}\label{der}
	\gamma(h_0\cdot h_0')=\gamma(h_0)\cdot h_0\ac\gamma(h_0'),\qquad \forall h_0,h_0'\in H_0.
\end{eqnarray} 
It is a semigroup with identity $\gamma_e(h_0)=e_{H_1}$  for the multiplication given by
\begin{eqnarray}\label{m of der}
	(\gamma_1\star\gamma_2)(h_0)=\gamma_1(\partial\gamma_2(h_0)\cdot h_0)\cdot\gamma_2(h_0).
\end{eqnarray}
Write $D(H_0,H_1)$ for the group of all units in $\Der(H_0,H_1)$.    We then obtain a group homomorphism
\begin{eqnarray}\label{Delta}
	\Delta:D(H_0,H_1)\to\Aut(H_1,H_0,\partial),\qquad 	\Delta(\gamma)(h_1,h_0)=(\gamma(\partial h_1)\cdot h_1, \partial\gamma(h_0)\cdot h_0),
\end{eqnarray}
and a group action of $\Aut(H_1,H_0,\partial)$ on $D(H_0,H_1)$ by automorphisms, defined by 
\begin{eqnarray}\label{ac of actor}
	(\epsilon,\rho)\ac_\mathcal{A}\gamma=\epsilon\circ\gamma\circ\rho^{-1}.
\end{eqnarray}
Then $\big(D(H_0,H_1)\xrightarrow{\Delta}\Aut(H_1,H_0,\partial)\big)$ is a crossed module, called the \textbf{actor crossed module} of $(H_1\xrightarrow{\partial}H_0)$, and denoted by $\mathcal{A}(H_1,H_0,\partial)$.

\begin{defi}[\cite{N}]
An \textbf{action} of a crossed module $(G_1\xrightarrow{\mu}G_0)$ on another crossed module $(H_1\xrightarrow{\partial}H_0)$ by automorphisms is a crossed module morphism $(\alpha,\beta):$ 
	\begin{eqnarray*}
		\begin{CD}
			G_1 @>\alpha>> D(H_0,H_1)\\
			@V\mu VV @VV\Delta V\\
			G_0 @>\beta>> \Aut(H_1,H_0,\partial).
		\end{CD}
	\end{eqnarray*}
\end{defi}
Denote by $\beta_1:G_0\to\Aut(H_1)$ and $\beta_0:G_0\to\Aut(H_0)$ the two components of $\beta$.

\begin{defi}\label{RBO of 2Gp}
     Let $(\alpha,\beta)$ be an action of a crossed module $(G_1\xrightarrow{\mu}G_0)$ on another crossed module $(H_1\xrightarrow{\partial}H_0)$. A \textbf{relative Rota-Baxter operator} on $(G_1\xrightarrow{\mu}G_0)$ with respect to the action $(\alpha,\beta)$ is a pair of smooth maps $(\B_1,\B_0)$ with $\B_i:H_i\to G_i$ such that $\mu\circ\B_1=\B_0\circ\partial$ and 
    \begin{itemize}
        \item[\rm(i)] $\B_1$ is a relative Rota-Baxter operator on the Lie group $G_1$ with respect to the action $\beta_1\circ\mu:G_1\to\Aut(H_1)$,
        \item [\rm(ii)]$\B_0$ is a relative Rota-Baxter operator on the Lie group $G_0$ with respect to the action $\beta_0:G_0\to\Aut(H_0)$,
        \item[\rm(iii)] $\B_0h_0\ac_G\B_1h_1=\B_1\big(h_0\ac_H\big(\beta_1(\B_0h_0)h_1\cdot_{H_1}\alpha(\B_0h_0\ac_G\B_1h_1)h_0^{-1}\big)\big)$, for all  $h_0\in H_0,h_1\in H_1$.
    \end{itemize}
\end{defi}

\begin{ex}
 A crossed module $(G_1\xrightarrow{\mu}G_0)$ acts on itself by the adjoint action \cite{N}, where the action map $(\alpha,\beta)$ from $(G_1\xrightarrow{\mu}G_0)$ to $\mathcal{A}(G_1,G_0,\mu)$ is defined by
\begin{eqnarray*}
	\alpha(g_1)g_0=g_1\cdot g_0\ac g_1^{-1},\qquad \beta(g_0)(g_1,g_0')=(g_0\ac g_1,g_0\cdot g_0'\cdot g_0^{-1}),\quad \forall g_0,g'_0\in G_0, g_1\in G_1.
\end{eqnarray*} 
Taking  $(H_1\xrightarrow{\partial}H_0)=(G_1\xrightarrow{\mu}G_0)$ and the action $(\alpha,\beta)$ as the adjoint action of $(G_1\xrightarrow{\mu}G_0)$ on itself in Definition \ref{RBO of 2Gp}, we recover  a \textbf{Rota-Baxter operator} $(\B_1, \B_0)$ on the crossed module $(G_1\xrightarrow{\mu}G_0)$ introduced in \cite{Jiang}. Namely, $\B_1$ and $\B_0$ are Rota-Baxter operators on $G_1$ and $G_0$, respectively, satisfying  $\mu\circ\B_1=\B_0\circ\mu$ and 
	\begin{eqnarray*}
		\B_0g_0\ac\B_1g_1=\B_1\big((g_0\cdot\B_0g_0)\ac(g_1\cdot\B_1g_1)\cdot\B_0g_0\ac(\B_1g_1)^{-1}\big), \qquad \forall g_0\in G_0,g_1\in G_1.
	\end{eqnarray*}
\end{ex}

We prove the equivalence of relative Rota-Baxter operators on Lie 2-groups and on crossed modules by establishing the following two propositions in sequence.
\begin{pro}\label{pro1}
If $(\B_1,\B_0)$ is a relative Rota-Baxter operator on a crossed module $(G_1\xrightarrow{\mu}G_0)$ with respect to an action $(\alpha,\beta)$ on $(H_1\xrightarrow{\partial}H_0)$, then there is a relative Rota-Baxter operator $(\B,\B_0)$ on the Lie 2-group $G_1\rtimes G_0\rightrightarrows G_0$ with respect to an action $(\phi^{\alpha \beta},\phi_0^{\alpha \beta})$ on the Lie 2-group  $H_1\rtimes H_0\rightrightarrows H_0$, where 
\begin{eqnarray}
	\label{tildeB}
	\B(h_1,h_0)&=&\big(\B_0h_0\ac_G\B_1(\beta_1(\B_0h_0)^{-1}(h_0^{-1}\ac_H h_1)),\B_0h_0\big),\\
	\nonumber
	\phi^{\alpha \beta}(g_1,g_0)(h_1,h_0)&=&(\beta_1(\mu g_1\cdot g_0)h_1\cdot\alpha(g_1)\beta_0(g_0)h_0,\beta_0(g_0)h_0),\\
	\nonumber
	\phi^{\alpha \beta}_0(g_0)h_0&=&\beta_0(g_0)h_0,\qquad \forall h_i\in H_i, g_i\in G_i,i=0,1.
\end{eqnarray}
\end{pro}
\begin{proof}
First, we show that $(\phi^{\alpha \beta},\phi_0^{\alpha \beta})$ is a Lie 2-group action. Following from \eqref{der}, \eqref{m of der}, \eqref{ac of actor} and the fact that $(\alpha,\beta)$ is a crossed module morphism, we have $\alpha(g_0\ac g_1)=\beta_1(g_0)\alpha(g_1)\beta_0(g_0^{-1})$ and 
\begin{eqnarray}
	\label{m of alpha}
		\alpha(g_1\cdot g_1')h_0&=&\alpha(g_1)(\partial\alpha(g_1')h_0\cdot h_0)\cdot\alpha(g_1')h_0\\ \nonumber
		&=&\alpha(g_1)\partial\alpha(g_1')h_0\cdot\alpha(g_1')h_0\cdot\alpha(g_1)h_0\\ \nonumber
		&=&\beta_1\mu(g_1)\alpha(g_1')h_0\cdot\alpha(g_1)h_0.
	\end{eqnarray}
Moreover, we obtain
	\begin{eqnarray*}
		&&\phi^{\alpha \beta}\big((g_1,g_0)\cdot(g_1',g_0')\big)(h_1,h_0)\\
		&=&\big(\beta_1(\mu g_1\cdot g_0\cdot\mu g_1'\cdot g_0')h_1\cdot\alpha(g_1\cdot g_0\ac g_1')\beta_0(g_0\cdot g_0')h_0,\beta_0(g_0\cdot g_0')h_0\big)\\
		&=&\big(\beta_1(\mu g_1\cdot g_0\cdot\mu g_1'\cdot g_0')h_1\cdot\beta_1(\mu g_1\cdot g_0)\alpha(g_1')\beta_0(g_0')h_0\cdot\alpha(g_1)\beta_0(g_0\cdot g_0')h_0,\beta_0(g_0\cdot g_0')h_0\big)\\
		&=&\phi^{\alpha \beta}(g_1,g_0)\big(\phi^{\alpha \beta}(g_1',g_0')(h_1,h_0)\big),
	\end{eqnarray*}
	which implies that $\phi^{\alpha \beta}$ is a Lie group action. Using again the fact that $(\alpha, \beta)$ is a crossed module morphism,  we have 	\begin{eqnarray*}
		\beta_1\mu(g_1)(h_0\ac h_1)&=&\beta_0\mu(g_1)h_0\ac\beta_1\mu(g_1)h_1\\&=&(\partial\alpha(g_1)h_0\cdot h_0)\ac\beta_1\mu(g_1)h_1\\&=&\alpha(g_1)h_0\cdot h_0\ac\beta_1\mu(g_1)h_1\cdot(\alpha(g_1)h_0)^{-1}.
	\end{eqnarray*}
Therefore,  we find
	\begin{eqnarray*}
		&&\phi^{\alpha \beta}(g_1,g_0)\big((h_1,h_0)\cdot(h_1',h_0')\big)\\
		&=&\big(\beta_1(\mu g_1\cdot g_0)(h_1\cdot h_0\ac h_1')\cdot\alpha(g_1)\beta_0(g_0)(h_0\cdot h_0'),\beta_0(g_0)(h_0\cdot h_0')\big)\\
		&=&\big(\beta_1(\mu g_1\cdot g_0)h_1\cdot \beta_1(\mu g_1\cdot g_0)(h_0\ac h_1')\cdot\alpha(g_1)\beta_0(g_0)h_0\cdot \beta_0(g_0)h_0\ac\alpha(g_1)\beta_0(g_0)h_0',\beta_0(g_0)(h_0\cdot h_0')\big)\\
		&=&\big(\beta_1(\mu g_1\cdot g_0)h_1\cdot\alpha(g_1)\beta_0(g_0)h_0\cdot\beta_0(g_0)h_0\ac(\beta_1(\mu g_1\cdot g_0)h_1'\cdot\alpha(g_1)\beta_0(g_0)h_0'),\beta_0(g_0)(h_0\cdot h_0')\big)\\
		&=&\phi^{\alpha \beta}(g_1,g_0)(h_1,h_0)\cdot\phi^{\alpha \beta}(g_1,g_0)(h_1',h_0').
	\end{eqnarray*}
	 Thus $\phi^{\alpha \beta}$ is a Lie group action by automorphisms.
	 
	  It is obvious that $\phi_0^{\alpha \beta}=\beta_0$ is a Lie group action by automorphisms. To verify that $(\phi^{\alpha \beta},\phi_0^{\alpha \beta})$ is a Lie groupoid morphism, we only check that $\phi^{\alpha \beta}$ preserves the groupoid multiplication. For $\big((g_1,g_0),(g_1',g_0')\big)\in(G_1\rtimes G_0)^{(2)}, \big((h_1,h_0),(h_1',h_0')\big)\in(H_1\rtimes H_0)^{(2)}$, we obtain
	 \begin{eqnarray*}
	 	&&\phi^{\alpha \beta}\big((g_1,g_0)*(g_1',g_0')\big)\big((h_1,h_0)*(h_1',h_0')\big)\\
	 	&=&\phi^{\alpha \beta}(g_1\cdot g_1',g_0')(h_1\cdot h_1',h_0')\\
	 	&=&\big(\beta_1(\mu g_1\cdot g_0)(h_1\cdot h_1')\cdot\alpha(g_1\cdot g_1')\beta_0(g_0')h_0',\beta_0(g_0')h_0'\big)\\
	 	&=&\big(\beta_1(\mu g_1\cdot g_0)(h_1\cdot h_1')\cdot\alpha(g_1)\beta_0(g_0)h_0'\cdot\alpha(g_1')\beta_0(g_0')h_0',\beta_0(g_0')h_0'\big)\\
	 	&=&\big(\beta_1(\mu g_1\cdot g_0)h_1\cdot\alpha(g_1)\partial\beta_1(g_0)h_1'\cdot\beta_1(g_0)h_1'\cdot\alpha(g_1)\beta_0(g_0)h_0'\cdot\alpha(g_1')\beta_0(g_0')h_0',\beta_0(g_0')h_0'\big)\\
	 	&=&\big(\beta_1(\mu g_1\cdot g_0)h_1\cdot\alpha(g_1)\beta_0(g_0)h_0\cdot\beta_1(g_0)h_1'\cdot\alpha(g_1')\beta_0(g_0')h_0',\beta_0(g_0')h_0'\big)\\
	 	&=&\phi^{\alpha \beta}(g_1,g_0)(h_1,h_0)*\phi^{\alpha \beta}(g_1',g_0')(h_1',h_0'),
	 \end{eqnarray*}
	  where we used \eqref{m of alpha} and the fact that $\alpha(g_1)$ is a derivation. Hence $(\phi^{\alpha \beta},\phi_0^{\alpha \beta})$ is a Lie 2-group action.

Then we verify that $\B:H_1\rtimes H_0\to G_1\rtimes G_0$ is a relative Rota-Baxter operator with respect to the Lie group action $\phi^{\alpha \beta}$. By definition, we see that $\B|_{H_i}=\B_i$. Therefore we only need to show the relation \eqref{rrbc} on crossed terms. In fact, we have
  \begin{eqnarray*}
  	\B(e,h_0)\cdot\B(h_1,e)=(\B_0h_0\ac\B_1h_1,\B_0h_0)=\B(h_0\ac\beta_1(\B_0h_0)h_1,h_0)=\B\big((e,h_0)\cdot\phi^{\alpha \beta}(\B(e,h_0))(h_1,e)\big).
   \end{eqnarray*}
  Using $\rm(iii)$ in Definition \ref{RBO of 2Gp} and the fact that $(\B_0h_0)^{-1}=\B_0\beta_0(\B_0h_0)^{-1}h_0^{-1}$, we have
  \begin{eqnarray*}
  	&&\B(h_1,e)\cdot\B(e,h_0)\\
  	&=&(e,\B_0h_0)\cdot(\B_0\beta_0(\B_0h_0)^{-1}h_0^{-1}\ac\B_1h_1,e)\\
  	&=&(e,\B_0h_0)\cdot\big(\B_1(\beta_0(\B_0h_0)^{-1}h_0^{-1}\ac(\beta_1(\B_0h_0)^{-1}h_1\cdot\alpha((\B_0h_0)^{-1}\ac\B_1h_1)\beta_0(\B_0h_0)^{-1}h_0),e\big)\\
  	&=&(e,\B_0h_0)\cdot(\B_1\beta_1(\B_0h_0)^{-1}(h_0^{-1}\ac(h_1\cdot\alpha(\B_1h_1)h_0)),e)\\
  	&=&\B\big((h_1,e)\cdot\phi^{\alpha \beta}(\B(h_1,e))(e,h_0)\big).	
  \end{eqnarray*}
    This proves that $\B$ is a relative Rota-Baxter operator. It suffices to show that $(\B,\B_0)$ is a Lie groupoid morphism. We only verify that $\B$ preserves the groupoid multiplication.
    For $\big((h_1, h_0),(h_1',h_0')\big)\in(H_1\rtimes H_0)^{(2)}$, we obtain
    \begin{eqnarray*}
    	\B_0h_0=\B_0(\partial h_1'\cdot h_0')=t\B(h_1',h_0')=\B_0h_0'\cdot\mu\B_1(\beta_1(\B_0h_0')^{-1}(h_0'^{-1}\ac h_1')).
    \end{eqnarray*}
    Write $\B_1(\beta_1(\B_0h_0')^{-1}(h_0'^{-1}\ac h_1'))$ for $\wp$ in short, and we deduce
    \begin{eqnarray*}
    	&&\B(h_1,h_0)*\B(h_1',h_0')\\
    	&=&\big(\B_0h_0\ac\B_1\beta_1(\B_0h_0)^{-1}(h_0^{-1}\ac h_1)\cdot\B_0h_0'\ac\B_1\beta_1(\B_0h_0')^{-1}(h_0'^{-1}\ac h_1'),\B_0h_0'\big)\\
    	&=&\big(\B_0h_0'\ac(\mu\wp\ac\B_1\beta_1(\mu \wp^{-1}\cdot(\B_0h_0')^{-1})(h_0'^{-1}\ac(\partial h_1'^{-1}\ac h_1))\cdot \wp),\B_0h_0'\big)\\
    	&=&\big(\B_0h_0'\ac(\wp\cdot\B_1\beta_1(\mu \wp^{-1}\cdot(\B_0h_0')^{-1})(h_0'^{-1}\ac( h_1'^{-1}\cdot h_1\cdot h_1'))),\B_0h_0'\big)\\
    	&=&\big(\B_0h_0'\ac\B_1\beta_1(\B_0h_0')^{-1}(h_0'^{-1}\ac(h_1\cdot h_1')),\B_0h_0'\big)\\
    	&=&\B((h_1,h_0)*(h_1',h_0')),
    \end{eqnarray*}
    where in the last second equation, we expanded $\wp$ and used the relation \eqref{rrbc} of $\B_1$.
    \end{proof}

\begin{pro}\label{pro2}
If $(\B,\B_0)$ is a relative Rota-Baxter operator on a Lie 2-group  $P\rightrightarrows P_0$ with respect to an action $(\phi,\phi_0)$ on $Q\rightrightarrows Q_0$, then $(\B|_{\ker s_Q},\B_0)$ is a relative Rota-Baxter operator on the associated crossed module $(\ker s_P\xrightarrow{t_P} P_0)$ with respect to an action $(\alpha,\beta)$ on $(\ker s_Q\xrightarrow{t_Q} Q_0)$:
\begin{eqnarray*}
			\begin{CD}
				\ker s_P @>\alpha>> D(Q_0,\ker s_Q)\\
				@V t_P VV @VV\Delta V\\
				P_0 @>\beta>> \Aut(\ker s_Q,Q_0,t_Q),
			\end{CD}
\end{eqnarray*}
which is defined by 
	\begin{eqnarray*}
		\alpha(p)q_0=\phi(p)\iota q_0\cdot \iota q_0^{-1},\qquad \beta(p_0)(q,q_0)=(\phi (\iota p_0)q,\phi_0(p_0)q_0), 	\end{eqnarray*}
for all $p_0\in P_0,p\in \ker s_P,q_0\in Q_0,q\in\ker s_Q.$
\end{pro}
\begin{proof}
	 To simplify the notation, we omit $\iota$ in the following proof. First, we show $\alpha$ and $\beta$ are well-defined. We only check that $\alpha(p)$ is a derivation from $Q_0$ to $\ker s_Q$. By definition, 
	\begin{eqnarray*}
		\alpha(p)(q_0\cdot q_0')&=&\phi(p)(q_0\cdot q_0')\cdot(q_0\cdot q_0')^{-1}=\phi(p)q_0\cdot q_0^{-1}\cdot q_0\ac(\phi(p)\ q_0'\cdot q_0'^{-1})\\
		&=&\alpha(p)q_0\cdot q_0\ac \alpha(p)q_0'.
	\end{eqnarray*}
It is clear that $\beta$ is a Lie group homomorphism. Then we verify that $\alpha$ is a Lie group homomorphism. Notice that for any $q,q'\in\ker s_Q$, we have
\begin{eqnarray}\label{gpm-gpdm}
	(q\cdot t q')* q'=q\cdot(tq'* q')=q\cdot q'.
\end{eqnarray}
Then by \eqref{phi}, \eqref{m of der} and \eqref{gpm-gpdm}, we find 
	\begin{eqnarray*}
		(\alpha(p)\star \alpha(p'))q_0&=&
		\big(\phi(p\cdot tp')q_0\cdot q_0^{-1}\big)*(\phi(p') q_0\cdot q_0^{-1})\\
		&=&(\phi(p\cdot tp')q_0*\phi(p')q_0)\cdot(q_0^{-1}* q_0^{-1})\\
		&=&\phi(p\cdot p')q_0\cdot q_0^{-1}=\alpha(p\cdot p')q_0.
	\end{eqnarray*}
	 Now we check that $\Delta\circ \alpha=\beta\circ t_P$. In fact, by \eqref{gpm-gpdm}, we have
	\begin{eqnarray*}
		\Delta(\alpha(p))(q,q_0)&=&(\phi(p) tq\cdot tq^{-1}\cdot q,\phi_0(tp)q_0)=(\phi(p)tq* q,\phi_0(tp)q_0)\\
		&=&(\phi(tp* p)(q* e),\phi_0(tp)q_0)=(\phi(tp)q,\phi_0(tp)q_0)\\&=&\beta(tp)(q,q_0).
	\end{eqnarray*}
	A direct verification shows that $\alpha(p_0\ac p)=\beta(p_0)\ac_\mathcal{A} \alpha(p)$.
	Therefore, $(\alpha,\beta)$ is an action of crossed modules.
 It follows from $\phi(p)q=\phi(\iota tp)q$ for $ p\in\ker s_P,q\in\ker s_Q$ that $\B|_{\ker s_Q}:\ker s_Q\to\ker s_P$ is a relative Rota-Baxter operator with respect to the Lie group action $\phi\iota\circ t_P$. It remains to see that $(\B|_{\ker s_Q},\B_0)$ satisfies the relation $\rm(iii)$ in Definition \ref{RBO of 2Gp}. For all $q_0\in Q_0,q\in\ker s_Q$, we have
\begin{eqnarray*}
	\B_0q_0\ac\B q&=&\iota\B_0q_0\cdot\B q\cdot\iota(\B_0q_0)^{-1}\\&=&\B\iota q_0\cdot\B q\cdot\B\iota(\phi_0(\B_0q_0)^{-1}q_0^{-1})\\
	&=&\B(\iota q_0\cdot\phi(\B\iota q_0)q\cdot\phi(\B\iota q_0\cdot\B q)\iota(\phi_0(\B_0q_0)^{-1}q_0^{-1})\\
	&=&\B\big(q_0\ac(\phi\iota(\B_0q_0)q\cdot\phi(\B_0q_0\ac\B q)\iota q_0^{-1}\cdot\iota q_0)\big)\\&=&\B\big(q_0\ac( (\beta_1(\B_0q_0)q)\cdot \alpha(\B_0q_0\ac\B q)q_0^{-1})\big),
\end{eqnarray*}
where in the last equation, $\beta_1$ denotes the first component of $\beta$. This completes the proof.
\end{proof}

Now we are in the position to present the main result of this subsection. 
\begin{thm}\label{bj}
   There is a bijection between relative Rota-Baxter operators on Lie 2-groups and relative Rota-Baxter operators on Lie group crossed modules.
\end{thm}
\begin{proof}
It suffices to show that the two processes in Proposition \ref{pro1} and Proposition \ref{pro2} are inverse to each other up to the Lie group isomorphisms: 
	\begin{eqnarray*}
		\ker s_P\rtimes P_0\xrightarrow {\pi_P} P, \quad (p,p_0)\mapsto p\cdot\iota p_0; \qquad 	\ker s_Q\rtimes Q_0\xrightarrow {\pi_Q} Q, \quad (q,q_0)\mapsto q\cdot\iota q_0.\end{eqnarray*}
	Explicitly, starting from a relative Rota-Baxter operator $(\B,\B_0)$ on a Lie 2-group $P\rightrightarrows P_0$ with respect to an action $(\phi,\phi_0)$ on $Q\rightrightarrows Q_0$, we have 
	\begin{eqnarray*}
		(\B,\B_0)\text{ w.r.t. } (\phi,\phi_0)\xrightarrow{\text{Prop. } \ref{pro2}}(\B|_{\ker s_Q},\B_0)\text{ w.r.t. }(\alpha,\beta)\xrightarrow{\text{Prop. } \ref{pro1}}(\tilde{\B},\B_0)\text{ w.r.t. }(\phi^{\alpha\beta},\phi_0^{\alpha\beta}).
	\end{eqnarray*}
	 We verify that $(\phi\circ(\pi_P\times\pi_Q),\phi_0)=(\pi_Q\circ\phi^{\alpha\beta},\phi_0^{\alpha\beta})$ and $
	 \pi_P\circ\tilde{\B}=\B\circ\pi_Q$. In fact, since $\phi(\iota tp)q=\phi(p)q$, we have 
	\begin{eqnarray*}
		\big(\pi_Q(\phi^{\alpha\beta}(p,p_0)(q,q_0)),\phi_0^{\alpha\beta}(p_0)q_0\big)&=&\big(\phi(\iota tp\cdot \iota p_0)q\cdot\phi(p\cdot \iota p_0)\iota q_0,\phi_0(p_0)q_0\big)\\
		&=&(\phi(p\cdot\iota p_0)(q\cdot \iota q_0),\phi_0(p_0)q_0)\\
		&=&\big(\phi(\pi_P(p,p_0))\pi_Q(q,q_0),\phi_0(p_0)q_0\big).
	\end{eqnarray*}
As $\B$ is a relative Rota-Baxter operator, we have 
    \begin{eqnarray*}
        \pi_P\tilde{\B}(q,q_0)
        &=&\B_0q_0\ac\B\phi \iota(\B_0q_0)^{-1}(q_0^{-1}\ac q)\cdot\iota\B_0q_0\\
        &=&\B\iota q_0\cdot\B\phi (\B\iota q_0)^{-1}(\iota q_0^{-1}\cdot q\cdot\iota q_0)\\
        &=&\B(q\cdot\iota q_0)=\B\pi_Q(q,q_0),
    \end{eqnarray*}
    which completes the proof. 
\end{proof}

\begin{rmk}
\begin{itemize}
	\item [\rm(i)] The above results indicate that Lie 2-group actions are in one-to-one correspondence with actions of crossed modules, which is stated without proof in \cite{LL}.
	\item [\rm(ii)] If $(\B_1,\B_0)$ is a Rota-Baxter operator on a crossed module $(G_1\xrightarrow{\mu} G_0)$ and $\B$ is defined by \eqref{tildeB} with $(\alpha,\beta)$ being the adjoint action, then $(G_1\rtimes G_0\rightrightarrows G_0,(\B,\B_0))$ is a Rota-Baxter Lie 2-group, which recovers \cite[Proposition 3.4]{Jiang}.
\end{itemize}	
\end{rmk}

\subsection{Properties and examples}
In the framework of crossed modules, we also have results analogous to Theorems \ref{graph} and \ref{descendt 2grp}, Proposition \ref{iso rRB} and Corollary \ref{corRB}. They can be derived either by applying Theorem \ref{bj} or by direct verification. The proofs are omitted here.

Let $(\alpha,\beta)$ be an action of a crossed module $(G_1\xrightarrow{\mu}G_0)$ on another crossed module $(H_1\xrightarrow{\partial}H_0)$. There is a \textbf{semi-direct product crossed module} \cite{N}:
\begin{eqnarray*}
	(H_1\rtimes G_1\xrightarrow{(\partial,\mu)}H_0\rtimes G_0), \qquad (\partial,\mu)(h_1,g_1)=(\partial h_1,\mu g_1),
\end{eqnarray*}
where $H_0\rtimes G_0$ and $H_1\rtimes G_1$ are semi-direct product Lie groups given by
\begin{eqnarray*}
	(h_0,g_0)\cdot(h_0',g_0')&=&(h_0\cdot\beta_0(g_0)h_0',g_0\cdot g_0'),\qquad \forall g_0,g_0'\in G_0,h_0,h_0'\in H_0,\\
	(h_1,g_1)\cdot(h_1',g_1')&=&(h_1\cdot\beta_1(\mu g_1)h_1',g_1\cdot g_1'),\qquad \forall g_1,g_1'\in G_1,h_1,h_1'\in H_1.
\end{eqnarray*}
The action $\unrhd:H_0\rtimes G_0\to\Aut(H_1\rtimes G_1)$ of $H_0\rtimes G_0$ on $H_1\rtimes G_1$ by automorphisms is defined by
	\begin{eqnarray*}
	(h_0,g_0)\unrhd(h_1,g_1)=(h_0\ac(\beta_1(g_0)h_1\cdot\alpha(g_0\ac g_1)h_0^{-1}),g_0\ac g_1). 
	\end{eqnarray*}
\begin{cor}
     Suppose that $(\alpha,\beta)$ is an action of a crossed module $(G_1\xrightarrow{\mu}G_0)$ on another crossed module $(H_1\xrightarrow{\partial}H_0)$. Then the following statements are equivalent:
     \begin{itemize}
     	\item[\rm(i)] The pair $(\B_1,\B_0)$ is a relative Rota-Baxter operator on $(G_1\xrightarrow{\mu}G_0)$ with respect to $(\alpha,\beta)$;
     	\item[\rm(ii)] The graph $(\mathrm{Gr}(\B_1)\xrightarrow{(\partial,\mu)}\Gr(\B_0))$ is a sub-crossed module of the semi-direct product crossed module $(H_1\rtimes G_1\xrightarrow{(\partial,\mu)}H_0\rtimes G_0)$;
     	\item[\rm(iii)] The pair
     	 $(\hat{\B}_1,\hat{\B}_0)$ with $\hat{\B}_1:H_1\rtimes G_1\to H_1\rtimes G_1$ and $\hat{\B}_0:H_0\rtimes G_0\to H_0\rtimes G_0$ defined by
     	\begin{eqnarray*}
     		\hat{\B}_1(h_1,g_1)=(e_{H_1},g_1^{-1}\cdot\B_1h_1),\qquad
     		\hat{\B}_0(h_0,g_0)=(e_{H_0},g_0^{-1}\cdot\B_0h_0),
     	\end{eqnarray*}
     	is a Rota-Baxter operator on the crossed module $(H_1\rtimes G_1\xrightarrow{(\partial,\mu)}H_0\rtimes G_0)$.
     \end{itemize}
\end{cor}

Observe that the space of smooth maps from $G_0$ to $G_1$ with the multiplication \eqref{m of der} is a semigroup. The group of its units is denoted by $\mathrm{Map}(G_0,G_1)$. Denote by
\begin{eqnarray*}
	\Diff(G_1,G_0,\mu)=\{(\sigma_1,\sigma_0)\in\Diff(G_1)\times\Diff(G_0)|  \mu\circ\sigma_1=\sigma_0\circ\mu\}.
\end{eqnarray*}
Then $\big(\mathrm{Map}(G_0,G_1)\xrightarrow{\Delta}\Diff(G_1,G_0,\mu)\big)$ is a crossed module, where $\Delta$ and the action are defined in the same fashion as in  \eqref{Delta} and \eqref{ac of actor}.

\begin{cor}\label{descendt cm}
    If $(\B_1,\B_0)$ is a relative Rota-Baxter operator on $(G_1\xrightarrow{\mu}G_0)$ with respect to an action $(\alpha,\beta)$ of $(G_1\xrightarrow{\mu}G_0)$ on $(H_1\xrightarrow{\partial}H_0)$, then
    \begin{itemize}
    	\item[\rm(i)]  $(H_1^\B\xrightarrow{\partial}H_0^\B)$ is a crossed module, called the \textbf{descendant crossed module}, where
    	\begin{eqnarray*}
    		h_1\cdot_{\B} h_1'&=&h_1\cdot\beta_1(\mu\B_1h_1)h_1',\qquad
    		h_0\cdot_{\B_0} h_0'=h_0\cdot\beta_0(\B_0h_0)h_0',\\
    		h_0\ac_{\B} h_1&=&h_0\ac_H(\beta_1(\B_0h_0)h_1\cdot\alpha(\B_0h_0\ac_G\B_1h_1)h_0^{-1}),\qquad \forall h_0,h_0'\in H_0, h_1,h_1'\in H_1.
    	\end{eqnarray*}
    	Moreover, $(\B_1,\B_0)$ is a crossed module morphism from $(H_1^\B\xrightarrow{\partial}H_0^\B)$ to $(G_1\xrightarrow{\mu}G_0)$.
    	\item[\rm(ii)] There is a crossed module morphism $(\theta,\tau)$ from the descendant crossed module $(H_1^\B\xrightarrow{\partial}H_0^\B)$ to $\big(\mathrm{Map}(G_0,G_1)\xrightarrow{\Delta}\Diff(G_1,G_0,\mu)\big)$, defined by
    	\begin{eqnarray*}
    		\theta(h_1)g_0&=&(\B_1\beta_1(g_0)h_1^\dagger)^{-1}\cdot g_0\ac\B_1h_1^\dagger,\\
    		\tau_0(h_0)g_0&=&(\B_0\beta_0(g_0)h_0^\dagger)^{-1}\cdot g_0\cdot\B_0h_0^\dagger,\\
    		\tau_1(h_0)g_1&=&(\B_1(h_0\ac_H\beta_1(\B_0h_0)\alpha(g_1)h_0^\dagger))^{-1}\cdot\B_0h_0\ac_G g_1.
    	\end{eqnarray*}
    Here $\tau_1:H_0^\B\to\Diff(G_1)$ and $\tau_0:H_0^\B\to\Diff(G_0)$ are the two components of $\tau$.
    	\item[\rm(iii)] If $(\B_1,\B_0)$ is a Rota-Baxter operator on the crossed module $(G_1\xrightarrow{\mu}G_0)$, then it is also a Rota-Baxter operator on $(G_1^\B\xrightarrow{\mu}G_0^\B)$.
    \end{itemize} 
\end{cor}

\begin{cor}
     Let $(\alpha,\beta)$ be an action of a crossed module $(G_1\xrightarrow{\mu}G_0)$ on another crossed module $(H_1\xrightarrow{\partial}H_0)$ and $(\B_1,\B_0)$ be a relative Rota-Baxter operator with respect to $(\alpha,\beta)$. If  $\phi_i:G_i\to G_i,\psi_i:H_i\to H_i,i=0,1$ are smooth maps such that $(\psi_1\times\phi_1,\psi_0\times\phi_0)\in\Aut\big(H_1\rtimes G_1, H_0\rtimes G_0, (\partial,\mu)\big)$, then 
     \begin{eqnarray*}
     	(\phi_1^{-1}\circ\B_1\circ\psi_1,\phi_0^{-1}\circ\B_0\circ\psi_0)
     \end{eqnarray*}
  is also a relative Rota-Baxter operator on $(G_1\xrightarrow{\mu}G_0)$  with respect to $(\alpha,\beta)$.
  
   In particular, if $(\B_1,\B_0)$ is a Rota-Baxter operator on  $(G_1\xrightarrow{\mu}G_0)$ and  $(\phi_1,\phi_0)\in\Aut(G_1,G_0,\mu)$, then $(\phi_1^{-1}\circ\B_1\circ\phi_1,\phi_0^{-1}\circ\B_0\circ\phi_0)$ is also a Rota-Baxter operator on $(G_1\xrightarrow{\mu}G_0)$.
\end{cor}

\begin{ex}
Suppose that $\gamma$ is a derivation in $D(G_0,G_1)$. Then $((\phi_1^{\gamma})^{-1}\circ\B_1\circ\phi_1^\gamma,(\phi_0^{\gamma})^{-1}\circ\B_0\circ\phi_0^\gamma)$ is a Rota-Baxter operator on $(G_1\xrightarrow{\mu}G_0)$, where $(\phi_1^\gamma,\phi_0^\gamma) \in\Aut(G_1,G_0,\mu)$ is defined by 
     \begin{eqnarray*}
        \phi_1^{\gamma}(g_1)=\gamma(\mu g_1)\cdot g_1,\qquad \phi_0^{\gamma}(g_0)=\mu\gamma(g_0)\cdot g_0.
    \end{eqnarray*}
\end{ex}

\begin{ex}
	A crossed module morphism is a relative Rota-Baxter operator with respect to the trivial action of crossed modules.
\end{ex}

\begin{ex}
	If $(H_1\xrightarrow{\partial}H_0)=(V_1\stackrel{\partial}{\rightarrow}V_0)$, is a 2-vector space, then $(\B_1,\B_0)$ with $\B_i: V_i\to G_i$ is a relative Rota-Baxter operator on the crossed module $(G_1\xrightarrow{\mu} G_0)$ with respect to an action $(\alpha, \beta)$ on $(V_1\stackrel{\partial}{\rightarrow}V_0)$ if $\B_1$ and $\B_0$ satisfy the following relations:
	\begin{eqnarray*}
		\B_0v_0\cdot\B_0v_0'&=&\B_0(v_0+\beta_0(\B_0v_0)v_0'),\\
		\B_1v_1\cdot\B_1v_1'&=&\B_1(v_1+v_1'+\alpha(\B_1v_1)\partial v_1'),\\
		\B_0v_0\ac\B_1v_1&=&\B_1\big(\beta_1(\B_0v_0)v_1-\alpha(\B_0v_0\ac_G\B_1v_1)v_0\big).
	\end{eqnarray*}
Comparing with Example \ref{abelian}, we see the advantage of the Lie 2-group framework.
\end{ex}

\begin{ex}
A relative Rota-Baxter operator $\B:H\to G$ on a Lie group $G$ with respect to an action $\phi:G\to\Aut(H)$ induces two relative Rota-Baxter operators on crossed modules: 
$(\B|_{e_H},\B)$ on the crossed module $(\{e_G\}\hookrightarrow G)$ with respect to the action $(\alpha,\beta)$ on $(\{e_H\}\hookrightarrow H)$, where
\begin{eqnarray*}
	\alpha(e_G)h=e_H,\quad \beta_1(g)e_H=e_H, \quad\beta_0(g)h=\phi(g)h,
\end{eqnarray*}
and $(\B,\B)$ on the crossed module $(G\xrightarrow{\Id}G)$ with respect to the action $(\alpha',\beta')$ on $(H\xrightarrow{\Id}H)$, where
\begin{eqnarray*}
	\alpha'(g)h=\phi(g)h\cdot_H h^{-1},\quad \beta_1'=\beta_0'=\phi.
\end{eqnarray*}
This example explains the construction in Example \ref{2rb}.
\end{ex}

\begin{ex}
	For a finite-dimensional real vector space $V$ with an antisymmetric bilinear form $\omega:V\times V\to\Real$, we have the ‘Heisenberg Lie group’ $V\oplus\mathbb{R}$ with multiplication 
	\begin{eqnarray*}
		 (v_1,r_1)(v_2,r_2)=(v_1+v_2,r_1+r_2+\omega(v_1,v_2)).
	\end{eqnarray*}
	The projection further gives a crossed module $(V\oplus\Real\xrightarrow{\partial}V)$ \cite{BL} with the action $\ac$ of $V$ on $V\oplus\Real$: 
	\begin{eqnarray*}
		v_1\ac(v_2,r)=(v_2,r+2\omega(v_1,v_2)).
	\end{eqnarray*}
	A Lie group action $\phi:G\to \mathrm{Sp}(V,\omega)\subset\mathrm{GL}(V)$ preserving $\omega$ gives a crossed module action $(\alpha,\beta)$ of  $(G\xrightarrow{\Id}G)$ on $(V\oplus\Real\xrightarrow{\partial}V)$:
	\begin{eqnarray*}
	\beta_0(g)v=\phi(g)v, \quad \beta_1(g)(v,r)=(\phi(g)v,r),\quad \alpha(g)v=\big(\phi(g)v-v,\omega(v,\phi(g)v)\big).
	\end{eqnarray*}
Then the pair of maps $\B_1:V\oplus\mathbb{R}\to G$ and $\B_0:V\to G$ is a relative Rota-Baxter operator on $(G\xrightarrow{\Id}G)$ with respect to $(\alpha,\beta)$ if and only if $\B_0$ is a relative Rota-Baxter operator on the Lie group $G$ relative to $\phi$ and $\B_1(v,r)=\B_0v$.
\end{ex}

\subsection{Relative Rota-Baxter operators on Lie algebra crossed modules}
We consider the infinitesimals of relative Rota-Baxter operators on Lie group crossed modules.

Let $\g$ and $\h$ be Lie algebras with an action $\bar{\phi}:\g\to\Der(\h)$ of $\g$ on $\h$ by derivations. A linear map $B:\h\to\g$ is called a \textbf{relative Rota-Baxter operator} (of weight 1) on $\g$ with respect to the action $\bar\phi$ if 
\begin{eqnarray*}
    [Bu,Bv]_\g=B(\bar\phi(Bu)v-\bar\phi(Bv)u+[u,v]_\h),\qquad \forall u,v\in\h.
\end{eqnarray*} 

     A \textbf{Lie algebra crossed module} $(\h_1\xrightarrow{\bar{\partial}}\h_0)$ consists of two Lie algebras $\h_1$ and $\h_0$, a Lie algebra homomorphism $\bar{\partial}:\h_1\to \h_0$ and an action $\ac_\h:\h_0\to\Der(\h_1)$ of  $\h_0$ on $\h_1$ by derivations, such that
    \begin{eqnarray*}
        \bar{\partial}a\ac_\h b=[a,b]_{\h_1}, \qquad
        \bar{\partial}(u\ac_\h a)=[u,\bar{\partial}a]_{\h_0}, \qquad \forall u\in \h_0,a,b\in \h_1.
    \end{eqnarray*}
    Also, we have the \textbf{actor crossed module}  $\bar{\mathcal{A}}(\h_1,\h_0,\bar{\partial})=\big(\Der(\h_0,\h_1)\xrightarrow{\Delta}\Der(\h_1,\h_0,\bar\partial)\big)$.  Then an \textbf{action} of a Lie algebra crossed module $(\g_1\xrightarrow{\bar{\mu}}\g_0)$ on another crossed module $(\h_1\xrightarrow{\bar{\partial}}\h_0)$ is defined as a morphism from $(\g_1\xrightarrow{\bar{\mu}}\g_0)$ to  $\bar{\mathcal{A}}(\h_1,\h_0,\bar{\partial})$; see \cite{CL} for details.

\begin{defi}
     Let $(\bar{\alpha},\bar{\beta})$ be an action of a crossed module $(\g_1\xrightarrow{\bar{\mu}}\g_0)$ on another crossed module $(\h_1\xrightarrow{\bar{\partial}}\h_0)$. A \textbf{relative Rota-Baxter operator} on $(\g_1\xrightarrow{\bar{\mu}}\g_0)$ with respect to $(\bar{\alpha},\bar{\beta})$ is a pair of linear maps $(B_1,B_0)$ with $B_i:\h_i\to\g_i$ 
such that $\bar{\mu}\circ B_1=B_0\circ\bar{\partial}$ and
    \begin{itemize}
        \item[\rm(i)] $B_1$ is a relative Rota-Baxter operator on the Lie algebra $\g_1$ with respect to the action $\bar{\beta}_1\circ \bar{\mu}$,
        \item [\rm(ii)] $B_0$ is a relative Rota-Baxter operator on the Lie algebra $\g_0$ with respect to the action $\bar{\beta}_0$,
        \item[\rm(iii)] $B_0u\ac_\g B_1a=B_1(\bar{\beta}_1(B_0u)a-\bar{\alpha}(B_1a)u+u\ac_\h a)$ for all $u\in \h_0,a\in \h_1$.
    \end{itemize}
\end{defi}

\begin{rmk}
	This notion is a special case of relative Rota-Baxter operators on Lie 2-algebras introduced in \cite{LW} when both the Lie 2-algebras and the actions are taken to be strict.
\end{rmk}

\begin{thm}\label{diff}
If $(\B_1,\B_0)$ is a relative Rota-Baxter operator on the Lie group crossed module $(G_1\xrightarrow{\mu}G_0)$ with respect to an action $(\alpha,\beta)$ of $(G_1\xrightarrow{\mu}G_0)$ on $(H_1\xrightarrow{\partial}H_0)$, then $((\B_1)_{*e},(\B_0)_{*e})$ is a relative Rota-Baxter operator on the corresponding Lie algebra crossed module $(\g_1\xrightarrow{\bar{\mu}}\g_0)$ with respect to the infinitesimal action $(\bar{\alpha},\bar{\beta})$ of $(\g_1\xrightarrow{\bar{\mu}}\g_0)$ on $(\h_1\xrightarrow{\bar{\partial}}\h_0)$.
\end{thm}
\begin{proof}
	We write $B_1=(\B_1)_{*e}, B_0=(\B_0)_{*e}.$ By \cite[Theorem 4.1]{JSZ}, $B_1$ and $B_0$ are relative Rota-Baxter operators on the corresponding Lie algebras $\g_1$ and $\g_0$ with respect to $\bar\beta_1\circ \bar{\mu}$ and $\bar\beta_0$, respectively. Since $\mu\circ\B_1=\B_0\circ\partial$, it is apparent that $\bar\mu\circ B_1=B_0\circ \bar\partial$.
    Moreover, for all $u\in\h_0,a\in\h_1$, we have 
    \begin{eqnarray*}
        B_0u\ac_\g B_1a&=&\frac{d}{ds}\frac{d}{dt}|_{s=t=0}\B_0(\exp(su))\ac_G\B_1(\exp(ta))\\
        &=&\frac{d}{ds}\frac{d}{dt}|_{s=t=0}\B_1\big(\exp(su)\ac_H\beta_1(\B_0\exp(su))\exp(ta)\\
        &&\cdot_{H_1}\exp(su)\ac_H\alpha(\B_0\exp(su)\ac_G\B_1\exp(ta))(\exp(su))^{-1}\big)\\
        &=&B_1\big(\frac{d}{ds}\frac{d}{dt}|_{s=t=0}\exp(su)\ac_H\beta_1(\B_0\exp(su))\exp(ta)\\
        &&+\frac{d}{ds}\frac{d}{dt}|_{s=t=0}\exp(su)\ac_H\alpha(\B_0\exp(su)\ac_G\B_1\exp(ta))(\exp(su))^{-1}\big)\\
        &=&B_1\big(\bar\beta_1(B_0u)a+u\ac_\h a-\bar\alpha(B_1a)u\big).
    \end{eqnarray*}
    Therefore, $(B_1,B_0)$ is a relative Rota-Baxter operator on $(\g_1\xrightarrow{\bar{\mu}}\g_0)$ with respect to $(\bar{\alpha},\bar{\beta})$.
\end{proof}

\section{Crossed homomorphisms on Lie 2-groups}\label{S5}
Let $G$ and $H$ be Lie groups with an action $\phi:G\to\Aut(H)$ of $G$ on $H$ by automorphisms. A \textbf{crossed homomorphism} with respect to $\phi$ is a smooth map $\D:G\to H$ such that 
\begin{eqnarray}\label{cr homo}
    \D(g\cdot_Gg')=\D g\cdot_H\phi(g)\D g', \qquad \forall g,g'\in G.
\end{eqnarray}
It is also called a \textbf{1-cocycle} of $G$ with coefficients in $H$. If $\D$ is bijective,  it gives a set-theoretical solution of the Yang-Baxter equation \eqref{YBE} by \cite{LYZ}. 

Infinitesimally, let $\g$ and $\h$ be Lie algebras with an action $\bar{\phi}:\g\to\Der(\h)$ of $\g$ on $\h$ by derivations. A \textbf{crossed homomorphism} with respect to $\bar{\phi}$  is a linear map $D:\g\to \h$ such that 
\begin{eqnarray*}
	D[x,y]_\g=\bar{\phi}(x)Dy-\bar{\phi}(y)Dx+[Dx,Dy]_\h, \qquad \forall x,y\in \g.
\end{eqnarray*}

Now we introduce the notion of crossed homomorphisms on Lie 2-groups.

\begin{defi}\label{RDO of 2Gp}
     Let $(\phi,\phi_0)$ be a Lie 2-group action of $P\rightrightarrows P_0$ on another Lie 2-group  $Q\rightrightarrows Q_0$. A \textbf{crossed homomorphism} with respect to $(\phi,\phi_0)$ is a pair $(\D,\D_0)$ of smooth maps such that
     \begin{itemize}
         \item $\D:P\to Q$ is a  crossed homomorphism  with respect to the action $\phi:P\to\Aut(Q)$,
         \item $\D_0:P_0\to Q_0$ is  a  crossed homomorphism  with respect to the action $\phi_0:P_0\to\Aut(Q_0)$,
         \item $(\D,\D_0)$ is a Lie groupoid morphism from $P\rightrightarrows P_0$ to $Q\rightrightarrows Q_0$.
     \end{itemize}
\end{defi}

 Denote by $\Gr(\D )=\{(\D p,p)|\forall p\in P\}$ and $\Gr(\D_0)=\{(\D_0p_0,p_0)|\forall p_0\in P_0\}$ the graphs of smooth maps $\D $ and $\D_0$, respectively.
\begin{pro}
    Let $(\phi,\phi_0)$ be a Lie 2-group action of $P\rightrightarrows P_0$ on another Lie 2-group $Q\rightrightarrows Q_0$. The pair of smooth maps $(\D,\D_0)$ with $\D: P\to Q$ and $\D_0: P_0\to Q_0$ is a crossed homomorphism with respect to $(\phi,\phi_0)$ if and only if  $\Gr(\D )\rightrightarrows\Gr(\D_0)$ is a Lie 2-subgroup of the semi-direct product Lie 2-group $Q\rtimes P\rightrightarrows Q_0\rtimes P_0$ given by \eqref{sd1}-\eqref{sd3}.
\end{pro}

\begin{proof}
 We infer from \cite[Proposition 5.4]{JS} that $\D$ and $\D_0$ are crossed homomorphisms on Lie groups with respect to the actions $\phi:P\to\Aut(Q)$ and $\phi_0:P_0\to\Aut(Q_0)$ if and only if $\Gr(\D)$ and $\Gr(\D_0)$ are Lie subgroups of $Q\rtimes P$ and $Q_0\rtimes P_0$. Furthermore, for $(p,p')\in P^{(2)}$, we see that	
 \begin{eqnarray*}
 	(\D p,p)*(\D p',p')=(\D p*\D p',p* p')\in\Gr(\D)
 \end{eqnarray*}
 if and only if $\D p*\D p'=\D(p* p')$. Therefore, $\Gr(\D)\rightrightarrows\Gr(\D_0)$ is a Lie subgroupoid of $Q\rtimes P\rightrightarrows Q_0\rtimes P_0$ if and only if $(\D,\D_0)$ is a Lie groupoid morphism from $P\rightrightarrows P_0$ to $Q\rightrightarrows Q_0$. This completes the proof.
\end{proof}

\begin{pro}\label{rd-ro}
    If $(\D ,\D_0)$ is a crossed homomorphism with respect to a Lie 2-group action $(\phi,\phi_0)$ of $P\rightrightarrows P_0$ on $Q\rightrightarrows Q_0$, then $(\hat{\D},\hat{\D}_0)$ with $\hat{\D}:Q\rtimes P\to Q\rtimes P$ and $\hat{\D}_0:Q_0\rtimes P_0\to Q_0\rtimes P_0$ defined by
\begin{eqnarray*}
    \hat{\D}(q,p)=(\phi(p^{-1})(q^{-1}\cdot\D p),e),\qquad
    \hat{\D}_0(q_0,p_0)=(\phi_0(p_0^{-1})(q_0^{-1}\cdot\D_0p_0),e),
\end{eqnarray*}
is a Rota-Baxter operator on the semi-direct product Lie 2-group $Q\rtimes P\rightrightarrows Q_0\rtimes P_0$.
\end{pro}
\begin{proof}
    By the fact that $\D_0$ is a crossed homomorphism with respect to $\phi_0$, we have
    \begin{eqnarray*}
&&\hat{\D}_0\big((q_0,p_0)\cdot\hat{\D}_0(q_0,p_0)\cdot(q_0',p_0')\cdot(\hat{\D}_0(q_0,p_0))^{-1}\big)\\
    &=&\hat{\D}_0(\D_0p_0\cdot\phi_0(p_0)q_0'\cdot\phi_0(p_0\cdot p_0'\cdot p_0^{-1})(q_0^{-1}\cdot\D_0p_0)^{-1},p_0\cdot p_0')\\
   &=&\big(\phi_0(p_0^{-1})(q_0^{-1}\cdot\D_0p_0)\cdot\phi_0(p_0'^{-1})(q_0'^{-1}\cdot\D_0p_0'),e\big)\\
        &=&\hat{\D}_0(q_0,p_0)\cdot\hat{\D}_0(q_0',p_0'),
    \end{eqnarray*}
    which implies that $\hat{\D}_0$ is a Rota-Baxter operator on $Q_0\rtimes P_0$. Similarly, $\hat{\D}$ is a Rota-Baxter operator on $Q\rtimes P$. 
    Moreover, to see that $(\hat{\D},\hat{\D}_0)$ is a groupoid endomorphism, we only check that $\hat{\D}$ preserves the groupoid multiplication. Since $(\D,\D_0)$ is a Lie groupoid morphism, for all $(q,q')\in Q^{(2)},(p,p')\in P^{(2)}$, we have
    \begin{eqnarray*}
        \hat{\D}((q,p)*(q',p'))
        &=&\big(\phi(p^{-1}* p'^{-1})((q^{-1}* q'^{-1})\cdot(\D p* \D p')),e\big)\\
        &=&\big(\phi(p^{-1}* p'^{-1})((q^{-1}\cdot\D p)*(q'^{-1}\cdot\D p')),e\big)\\
        &=&\big(\phi(p^{-1})(q^{-1}\cdot\D p)*\phi(p'^{-1})(q'^{-1}\cdot\D p'),e\big)\\
        &=&\hat{\D}(q,p)*\hat{\D}(q',p').
    \end{eqnarray*}
    Thus $(\hat{\D},\hat{\D}_0)$ is a Rota-Baxter operator on $Q\rtimes P\rightrightarrows Q_0\rtimes P_0$.
\end{proof}

In \cite{GLST}, a crossed homomorphism $\D: G\to H$ on Lie groups with respect to an action $\phi:G\to \Aut(H)$ gives rise to another action of $G$ on $H$. Here we present a similar result for crossed homomorphisms on Lie 2-groups.

\begin{thm}
	If $(\D,\D_0)$ is a crossed homomorphism with respect to a Lie 2-group action $(\phi,\phi_0)$ of $P\rightrightarrows P_0$ on $Q\rightrightarrows Q_0$, then there is another Lie 2-group action $(\tilde{\phi},\tilde{\phi}_0)$ of $P\rightrightarrows P_0$ on $Q\rightrightarrows Q_0$, called the \textbf{derived action},
given by
	\begin{eqnarray*}
		\tilde{\phi}(p)q=\D p\cdot\phi(p)q\cdot(\D p)^{-1},\qquad
		\tilde{\phi}_0(p_0)q_0=\D_0p_0\cdot\phi_0(p_0)q_0\cdot(\D_0p_0)^{-1}.
	\end{eqnarray*}
	Moreover, $(\tilde{\D},\tilde{\D}_0)$ defined by 
	\begin{eqnarray*}
		\tilde{\D}:P\to Q,\qquad \tilde{\D}p=(\D p)^{-1},\qquad
		\tilde{\D}_0:P_0\to Q_0,\qquad \tilde{\D}_0p_0=(\D_0p_0)^{-1},
	\end{eqnarray*}
is a crossed homomorphism with respect to the derived action $(\tilde{\phi},\tilde{\phi}_0)$.
\end{thm}
\begin{proof}
	By \cite[Lemma 4.1]{GLST}, $\tilde{\phi}$ and $\tilde{\phi}_0$ are Lie group actions of $P$ on $Q$ and $P_0$ on $Q_0$ by automorphisms, $\tilde{\D}$ and $\tilde{\D}_0$ are crossed homomorphisms with respect to $\tilde{\phi}$ and $\tilde{\phi}_0$. To show that $(\tilde{\phi},\tilde{\phi}_0)$ is a groupoid morphism, we only verify that
		$\tilde{\phi}(p* p')(q* q')=\tilde{\phi}(p)q*\tilde{\phi}(p')q'.$
	By \eqref{phi} and the fact that groupoid multiplication is a Lie group homomorphism, we have 
	\begin{eqnarray*}
		\tilde{\phi}(p* p')(q* q')&=&\D(p* p')\cdot\phi(p* p')(q* q')\cdot(\D(p* p'))^{-1}\\
		&=&(\D p* \D p')\cdot(\phi(p)q*\phi(p')q')\cdot((\D p)^{-1}*(\D p')^{-1})\\
		&=&(\D p\cdot\phi(p)q\cdot(\D p)^{-1})*(\D p'\cdot\phi(p')q'\cdot(\D p')^{-1})\\
		&=&\tilde{\phi}(p)q*\tilde{\phi}(p')q',
	\end{eqnarray*}
	where we also used that $(\D,\D_0)$ is a Lie groupoid morphism.  By a simple verification, $(\tilde{\D},\tilde{\D}_0)$ is a Lie groupoid morphism since $(\D,\D_0)$ is a Lie groupoid morphism and groupoid multiplication is a Lie group homomorphism. Thus $(\tilde{\D},\tilde{\D}_0)$ is a crossed homomorphism with respect to $(\tilde{\phi},\tilde{\phi}_0)$.
\end{proof}

The notion of crossed homomorphisms on Lie algebra crossed modules can be viewed as a particular case of that in \cite{LW}.
Let $(\bar{\alpha},\bar{\beta})$ be an action of a crossed module $(\g_1\xrightarrow{\bar{\mu}}\g_0)$ on another crossed module $(\h_1\xrightarrow{\bar{\partial}}\h_0)$. A \textbf{crossed homomorphism} with respect to $(\bar{\alpha},\bar{\beta})$ is a pair of linear maps  $D_i:\g_i\to\h_i$ such that $D_1$ and $D_0$ are crossed homomorphisms with respect to the action $\bar{\beta}_1\circ \bar{\mu}$ and $\bar{\beta}_0$ satisfying 
\begin{eqnarray*}
	\bar{\partial}\circ D_1=D_0\circ\bar{\mu},\qquad
	D_1(x\ac_\g \xi)=\bar\beta_1(x)D_1\xi-\bar\alpha(\xi)D_0x+D_0x\ac_\h D_1\xi, \qquad \forall x\in \g_0,\xi\in \g_1.
\end{eqnarray*}
Here we present crossed homomorphisms on Lie group crossed modules.

\begin{defi}
     Let $(\alpha,\beta)$ be an action of a crossed module $(G_1\xrightarrow{\mu}G_0)$ on another  crossed module $(H_1\xrightarrow{\partial}H_0)$. A \textbf{crossed homomorphism} with respect to the action $(\alpha,\beta)$ is a pair of smooth maps $(\D_1,\D_0)$ with $\D_i:G_i\to H_i$ 
     such that $\partial\circ \D_1=\D_0\circ\mu$ and
    \begin{itemize}
        \item[\rm(i)] $\D_1$ is a crossed homomorphism with respect to the action $\beta_1\circ\mu:G_1\to\Aut(H_1)$,
         \item[\rm(ii)] $\D_0$ is a crossed homomorphism with respect to the action $\beta_0:G_0\to\Aut(H_0)$,
        \item[\rm(iii)] $\D_1(g_0\ac_G g_1)=\D_0g_0\ac_H \big(\beta_1(g_0)\D_1g_1\cdot\alpha(g_0\ac g_1)(\D_0g_0)^{-1}\big)$, for all  $g_0\in G_0,g_1\in G_1$.
    \end{itemize}
\end{defi}

\begin{thm}\label{inverse}
Let $(\alpha,\beta)$ be an action of a Lie group crossed module $(G_1\xrightarrow{\mu}G_0)$ on another Lie group crossed module $(H_1\xrightarrow{\partial}H_0)$. 
     \begin{itemize}
         \item[\rm(i)] The formal inverse of a relative Rota-Baxter operator on $(G_1\xrightarrow{\mu}G_0)$ with respect to $(\alpha,\beta)$ is a crossed homomorphism with respect to $(\alpha,\beta)$. 
         \item[\rm(ii)] If $(\D_1,\D_0)$ is a crossed homomorphism with respect to $(\alpha,\beta)$, then $((\D_1)_{*e},(\D_0)_{*e})$ is a crossed homomorphism with respect to the infinitesimal action $(\bar\alpha,\bar\beta)$.
     \end{itemize}
\end{thm}
\begin{proof}
    $\rm(i)$  Let $(\D_1 ,\D_0)$ be the inverse of a relative Rota-Baxter operator $(\B_1,\B_0)$ on $(G_1\xrightarrow{\mu}G_0)$ with respect to $(\alpha,\beta)$. It follows from a similar argument in \cite[Theorem 2.17]{GLS} that the maps $\D_1, \D_0$ as the inverses of $\B_1, \B_0$ are crossed homomorphisms on Lie groups with respect to $\beta_1\circ\mu$ and $\beta_0$. By $\mu\circ \B_1=\B_0\circ \partial$, it is direct that $\partial\circ\D_1 =\D_0\circ\mu$. Consider the condition \rm(iii) in Definition \ref{RBO of 2Gp}, i.e.,
    \begin{eqnarray*}
        \B_0h_0\ac\B_1h_1=\B_1\big(h_0\ac(\beta_1(\B_0h_0)h_1\cdot\alpha(\B_0h_0\ac\B_1h_1)h_0^{-1})\big).
    \end{eqnarray*}
    Replacing $\B_ih_i$ by $g_i$ (for $i=0,1$) leads to $h_i=\D_ig_i$. Then we apply $\D_1$ to both sides of the above condition and obtain
    \begin{eqnarray*}
        \D (g_0\ac g_1)=\D_0g_0\ac\big(\beta_1(g_0)\D g_1\cdot\alpha(g_0\ac g_1)(\D_0g_0)^{-1}\big).
    \end{eqnarray*}
    Thus $(\D_1,\D_0)$ is a crossed homomorphism with respect to $(\alpha,\beta)$. 
    
    $\rm(ii)$
    Write $D_1=(\D_1)_{*e}$ and $ D_0=(\D_0)_{*e}$. By a discussion similar to \cite[Theorem 2.17]{GLS}, we see $D_1$ and $D_0$ are crossed homomorphisms on the corresponding Lie algebras with respect to $\bar\beta_1\circ \bar\mu$ and $\bar\beta_0$. It follows from $\partial\circ\D_1 =\D_0\circ\mu$ that $\bar\partial\circ D_1=D_0\circ \bar\mu$. 
    For all $x\in\g_0,\xi\in\g_1$, we have 
    \begin{eqnarray*}
        D_1(x\ac_\g \xi)&=&\frac{d}{ds}\frac{d}{dt}|_{s=t=0}\D_1 (\exp(tx)\ac_G\exp(s\xi))\\
        &=&\frac{d}{ds}\frac{d}{dt}|_{s=t=0} \D_0\exp(tx)\ac_H\beta_1(\exp(tx))\D_1\exp(s\xi)\\
        &&+\frac{d}{ds}\frac{d}{dt}|_{s=t=0}\D_0\exp(tx)\ac_H\alpha(\exp(tx)\ac_G \exp(s\xi))(\D_0\exp(tx))^{-1}\\
        &=&\frac{d}{ds}\frac{d}{dt}|_{s=t=0} \exp(tD_0x)\ac_H\beta_1(\exp(tx))\exp(sD_1\xi)\\
        &&+\frac{d}{ds}\frac{d}{dt}|_{s=t=0}\exp(tD_0x)\ac_H\alpha(\exp(tx)\ac_G \exp(s\xi))\exp(-tD_0x)\\
        &=&D_0x\ac_\h D_1\xi+\bar\beta_1(x)D_1\xi-\bar\alpha(\xi)D_0x.
    \end{eqnarray*}
    Thus $(D_1,D_0)$ is a crossed homomorphism on Lie algebra crossed modules with respect to $(\bar\alpha,\bar\beta)$.
\end{proof}


\end{document}